\documentclass[12 pt,a4 paper]{article}
\usepackage{graphicx} 
\usepackage{amsfonts}
\usepackage{adjustbox}
\usepackage{amsmath}
\usepackage{cite}
\usepackage[margin=1in]{geometry}
\usepackage{algorithmic}
\usepackage{array}
\usepackage{graphicx} \usepackage{cite}
\usepackage{xcolor}
\usepackage{adjustbox}
\usepackage[margin=1in]{geometry}
\usepackage[ruled,vlined]{algorithm2e}
\usepackage{amsmath,amssymb,amsfonts,amsthm}
\usepackage{algorithmic}
\usepackage{textcomp}
\usepackage{float}
\usepackage{color}
\usepackage{soul}
\usepackage{enumerate}
\usepackage{setspace}
\usepackage [T1]{fontenc}
\usepackage{xcolor}
\usepackage{amsmath,amssymb,amsfonts,amsthm}
\usepackage{algorithmic}
\usepackage{textcomp}
\usepackage{float}
\usepackage{color}
\usepackage{soul}
\usepackage{enumerate}
\usepackage{setspace}
\usepackage [T1]{fontenc}
\newcommand{\bigO}{\mathcal{O}}
\newtheorem{definition}{Definition}[section]
\newtheorem{theorem}{Theorem}[section]
\newtheorem{lemma}{Lemma}[section]
\newtheorem{proposition}{Proposition}[section]
\newtheorem{corollary}{Corollary}[section]
\newtheorem{example}{Example}[section]
\newtheorem{remark}{Remark}[section]
\title{Big data searching using words}
\author{Santanu Acharjee$^1$ and  Ripunjoy Choudhury$^2$\\
$^{1,2}$Department of Mathematics\\
Gauhati University\\
Guwahati-781014, Assam, India.\\
e-mails: $^1$sacharjee326@gmail.com, $^2$ripunjoy07@gmail.com}
\date{}
\begin{document}
\onehalfspacing
\maketitle
\section*{Abstract}
Big data analytics is one of the most promising areas of new research and development in computer science, enterprises, e-commerce, and defense. For many organizations, big data is considered one of their most important strategic assets. This explosive growth has made it necessary to develop effective techniques for examining and analyzing big data from mathematical perspectives. Among various methods of analyzing big data, topological data analysis (TDA) is now considered one of the useful tools. However, there is no fundamental concept related to the topological structure in big data. In this paper, we present fundamental concepts related to the neighborhood structures of words in big data search, laying the groundwork for developing topological frameworks for big data in the future. We also introduce the notion of big data primal within the context of big data search and explore how neighborhood structures, combined with the Jaccard similarity coefficient, can be utilized to detect anomalies in search behavior.\\

\noindent
{\bf Keywords:} Big data, neighborhood, search space, graph, anomaly, computational complexity.\\
\noindent
{\bf 2020 AMS Classifications:} 68P05; 68P10; 94A16; 54A99.\\

\section{Introduction}
In the last few decades of the previous century, the world witnessed significant advancements in industry and technology. The emergence of two new fields, computer science and information science, during this period has significantly contributed to the rapid advancement of technology and business. New developments in the fields of data analysis and data gathering have begun concurrently with these two expansions. Data collection and analysis have long been foundational to data-based research. Traditionally, such approaches have relied on statistical methods to interpret data and derive relevant insights. The mid-1990s marked a turning point with the emergence of the Internet and the World Wide Web, fundamentally transforming the data landscape. Computer science was increasingly integrated with traditional manual methods to facilitate data generation from diverse sources. As a result, the digital revolution accelerated the evolution of data creation and collection methods at an unprecedented pace. According to \cite{1}, global data production, capture, copying, and consumption were expected to increase rapidly, reaching a projected total of 64.2 zettabytes in 2020. It was further anticipated that, after five additional years of growth, the volume of data created worldwide would surpass 180 zettabytes by 2025 \cite{1}. In 2020, the amount of data generated and duplicated reached a record high \cite{1}. This surge was largely driven by the increased demand caused by the COVID-19 pandemic, as more individuals worked and studied from home and relied more heavily on home entertainment options \cite{1}. \\

\noindent
Recently, big data analytics has created numerous opportunities for researchers in mathematics and computer science, market analysts, and decision-makers in multinational corporations \cite{2}. Its significance is also evident in the defense sector \cite{3}. In the field of biology, big data plays a crucial role in DNA-based research \cite{4}. The applications of deep learning \cite{5}, transfer learning \cite{6}, reinforcement learning \cite{5}, and other related approaches in big data analytics are rapidly expanding to enhance decision-making capabilities while managing large datasets. Globally, various sectors—such as healthcare, agriculture, tourism, energy markets, and stock markets—have greatly benefited from the integration of machine learning and big data technologies. In Europe, the applications of these technologies have become increasingly significant, particularly in energy markets \cite{7}, real estate markets \cite{8}, and other key industries. The European tourism industry has also experienced notable advancements through the adoption of machine learning and big data technologies \cite{9}. In Asia, researchers have widely applied these technologies across various domains. For instance, neural networks and machine learning models have been employed to predict the prices of agricultural commodities in Southeast Asia \cite{10, 11}. A series of recent studies by Jin and Xu in China \cite{12, 13, 14, 15, 16, 17, 18} illustrate the effective use of neural networks and machine learning techniques to forecast the prices of key commodities such as green gram, crude oil, carbon emission allowances, and regional steel in North China. These studies have significantly influenced our own research direction. In summary, big data analytics is poised to serve as a transformative force across a wide array of domains in the 21st century.\\

\noindent
  Topological Data Analysis (TDA) is clearly and accessibly explained for a general audience by Knudson \cite{19}. Among the various approaches to big data analysis, Topological Data Analysis (TDA) holds particular significance. Traditional statistical techniques, such as regression analysis, are most effective when applied to dispersed data \cite{3}. However, in scenarios where data points are distributed along geometric shapes in two-dimensional space, regression analysis becomes inadequate due to its reliance on a linear regression line \cite{3}. This limitation becomes even more pronounced when data is dispersed in higher-dimensional spaces, making it challenging to detect underlying geometric structures using conventional methods. To address such challenges in big data analysis, the application of Algebraic Topology, a branch of mathematics, becomes essential \cite{20}. \\

\noindent
Regarding the theoretical underpinnings of big data analytics, Coveney et al. \cite{21} emphasized the urgent need to develop robust theoretical frameworks, cautioning that the impact of big data would be diminished without such foundations. A similar viewpoint was expressed by Succi and Coveney \cite{22}, who argued for alternative theoretical models to support big data analytics.\\

\noindent
In response to these concerns, recent research has increasingly explored the application of Topological Data Analysis (TDA) as a novel lens for interpreting complex datasets. Notably, studies such as \cite{23, 24, 25} have examined the relationship between data and topology through Euclidean space $\mathbf{E^n}$. Offroy and Duponchel \cite{26} applied TDA to address challenges in biological, analytical, and physical chemistry data. Later, Sná\v{s}el et al. \cite{20} conducted a comprehensive survey in 2017, highlighting geometrical and topological approaches to big data and expressing optimism about the future applicability of topology in data analysis. More recently, Boyd et al. \cite{27} explored the use of TDA in geoscience, demonstrating that it offers a more nuanced and effective approach to analyzing high-dimensional data compared to traditional clustering techniques, and positioning it as a powerful tool for advancing quantitative research in geoscience education.\\

\noindent
Collectively, these developments suggest that the use of TDA for big data analysis is rapidly expanding. Just as every lock is manufactured to fit a specific key—designed with a structure that precisely accommodates only that key—TDA offers specialized analytical structures capable of unlocking complex patterns in high-dimensional data. This analogy raises important questions from the perspective of big data analytics:\\

\begin{enumerate}
    \item \textit{What are the hidden topological features in big data?}
    \item \textit{Can we establish generalized topological foundations for big data searching and big data analytics?} 
\end{enumerate}

\noindent
Although Topological Data Analysis (TDA) is grounded in the principles of algebraic topology, as noted in \cite{20}, there is a lack of foundational concepts directly addressing the hidden topological structures within big data analytics—apart from a few analytical procedures involving TDA, regression analysis, and related methods. Recently, Acharjee \cite{3} proposed preliminary ideas connecting topology and big data through a secret sharing scheme in the defense sector.\\

\noindent
Motivated by the arguments of Coveney et al. \cite{21} and Succi and Coveney \cite{22} regarding the need for alternative theoretical foundations in big data analytics, this study aims to establish a novel relationship between words based on their search spaces. Building on this idea, we introduce a topology-based big data search framework. The study explores several mathematically significant results arising from this system and highlights their novelty. Additionally, it examines graph-based structures representing word relationships in the context of big data search and proposes an anomaly detection method, complete with Python implementation and a relevant case study. Recognizing the role of data proximity, the study further introduces a modified primal structure aimed at improving search efficiency in big data environments.\\

\noindent
The structure of this article is organized to guide the reader through a logical progression of ideas and results. Section 2 offers a review of prior work, highlighting key contributions, identifying existing gaps, and presenting the innovations proposed in this study. Section 3 introduces essential definitions and foundational results that form the basis for the subsequent analysis. Section 4 presents the concept of word relationships derived from search spaces, along with significant theoretical findings. In Section 5, we develop the notion of neighborhood structures of words and introduce several new results. Sections 6 and 7 focus on the construction of a word graph for big data search and an anomaly detection method, respectively. Section 8 proposes a modified primal structure aimed at improving the efficiency and accuracy of big data search mechanisms. Section 9 evaluates the relevance and implications of the definitions and results introduced, discussing both the strengths and limitations of the proposed methodologies. In section 10, the computational complexity of the proposed structure is calculated. The comparison of our proposed method with the TF-IDF and cosine similarity methods is done in section 11. Moreover, limitations on empirical and quantitative validations are discussed in section 12. Section 13 consists of discussions related to our findings. Finally, the conclusion is added in section 14.
\noindent
\section{A brief overview of previous studies:}
The term ``Big Data" was first introduced in 1997 by Cox and Ellsworth of NASA \cite{28}. In their work, they mentioned, ``{\it visualization provides an interesting challenge for computer systems: data sets are generally quite large, taxing the capacities of main memory, local disk, and even remote disk. We call this the problem of big data.}” \cite{28}. Later, Chen and Zhang \cite{29} offered a comprehensive definition of big data, describing it as datasets that are difficult to collect, store, filter, exchange, analyze, and visualize without the support of modern technologies. In essence, big data refers to highly complex datasets, and Big Data Analytics is the discipline focused on the large-scale processing and interpretation of such data.\\

\noindent
Big data is often characterized by five distinct features, commonly referred to as the ``5 V’s of big data'': volume, value, velocity, variety, and veracity. Acharjee \cite{3} highlighted these five dimensions, emphasizing their role in making big data increasingly dynamic. Over time, additional ``V’s'' have been proposed to capture emerging characteristics, as noted in \cite{3,30}. Readers may refer to \cite{3,30} for a detailed explanation of each ``V.''\\

\noindent
In terms of mathematical foundations, Sun and Wang \cite{31} were among the first to explore theoretical aspects of big data search, and their work is regarded as a breakthrough in connecting mathematics with big data search problems. Despite their contribution, relatively few studies have since addressed the development of mathematical models for big data. Motivated by this gap, the present paper seeks to contribute novel mathematical concepts related to big data search, building on the ideas introduced in \cite{31}.\\

\noindent
Historically, big data search relied on traditional keyword-based indexing techniques, such as TF-IDF (Term Frequency–Inverse Document Frequency) \cite{32}, and vector-based similarity measures like cosine similarity \cite{33}, which focus on identifying exact or closely matching terms across documents. More recently, deep learning models such as Word2Vec \cite{34} and BERT \cite{35} have enabled more advanced semantic representation by embedding words into high-dimensional vector spaces, where semantically related words are placed closer together.\\

\noindent
In parallel, Topological Data Analysis (TDA) has demonstrated significant potential in fields such as biology and image processing, where it helps uncover hidden structures in complex datasets. However, the integration of TDA into text-based big data search remains largely unexplored. Most applications of TDA have been concentrated on clustering or pattern detection in non-textual domains, with minimal research focused on leveraging topological structures derived from textual data or word relationships. This presents a significant opportunity for advancing the theoretical and practical integration of TDA in the domain of big data search.\\

\noindent
Despite significant progress in big data analytics, several critical research gaps remain unaddressed. Traditional big data search methods often overlook topological concepts, which have the potential to reveal deeper relational structures within datasets. While Topological Data Analysis (TDA) has garnered increasing attention for its efficacy in analyzing complex data, there remains a noticeable absence of structured methodologies that apply topological frameworks specifically to word-based data search. Furthermore, current anomaly detection techniques are predominantly designed for numerical datasets and rely heavily on statistical or vector-space models. These methods fall short when applied to textual data, as they do not effectively capture anomalies that arise from contextual and topological relationships among words. Additionally, prior research has not introduced foundational structures or neighborhood models aimed at enhancing the structural representation of text within large-scale data environments. Addressing these limitations is essential for advancing the integration of topology into big data text analytics.\\

\noindent
In response to the identified research gaps, this study proposes several novel concepts and techniques. Central among them is the introduction of a neighborhood structure for words, offering a foundational framework to represent the proximity and relational dynamics between terms. Unlike conventional methods based on simple co-occurrence or semantic similarity, this approach establishes topological relationships that enable more meaningful clustering and retrieval of related words. Furthermore, by integrating the Jaccard similarity coefficient into this topological context, the study presents a robust technique for detecting anomalies and irregular patterns in word-based searches. This method advances existing anomaly detection models by focusing on the relative positioning and interconnections of words within the neighborhood structure, rather than relying solely on statistical deviations. Additionally, the paper introduces the concept of a big data primal structure, which decomposes large datasets into more manageable and interpretable components. This structure %not only improves computational efficiency but also 
serves as a foundational model designed to enhance the structural representation of textual data in big data environments.\\

\section{Preliminaries:} In this section, we adopt definitions from existing literature that will be used in the following sections.\\
\begin{definition}\cite{36}
 Let $\mathbf{R}$ be a binary relation on $U$, namely, $\mathbf{R}$ is a subset of
the Cartesian product $U\times U$. When $(x,y)\in \mathbf{R}$, we also write $x\mathbf{R}y$.  
\end{definition}
\begin{definition}\cite{36}
    The relation $\mathbf{R}$ is referred to as serial if for all $x\in U$ there exists $y\in U$ such that $x\mathbf{R}y$. 
\end{definition}

\begin{definition} \cite{37}
 The relation $\mathbf{R}$ is preorder if and only if $\mathbf{R}$ is reflexive and transitive. \end{definition}
 
In this section, we recall several definitions and key results from \cite{37} regarding topologies generated by a relation on a set.\\

\begin{definition}\cite{37}
            If $\mathbf{R}$ is a relation on $X$, then the afterset of $x\in X$ is $x\mathbf{R}$, where
$\mathbf{xR} = \{\,y : x\mathbf{R}y\,\}$ and the forset of $x\in X$ is $\mathbf{R}x$, where $\mathbf{R}x = \{\,y : y\mathbf{R}x\,\}$.
        \end{definition}
      \begin{example}
          Let $X=\{\,a,b,c,d\,\}$ and $\mathbf{R}$ be a relation on $X$ such that $\mathbf{R}=\{\,(a,a),(a,b),(c,a),(d,a)\,\}.$ Then, forset and afterset  of $a$ are $\mathbf{R}a=\{\,c,d,a\,\}$, and $a\mathbf{R}=\{\,b,a\,\}$ respectively.
      \end{example}
 \begin{proposition}\cite{37}
If $\mathbf{R}$ is a relation on $X$, then the class $S_{1} =\{\,x\mathbf{R}:x\in X\,\}(resp. S_{2} =
\{\,\mathbf{R}x : x \in X\,\})$ is a subbase for the topology $\tau_{1}(resp. \tau_{2})$ on $X$.
    \end{proposition}
  \begin{definition}\cite{37}
     If $\tau$ is a topology in a finite set $X$ and the class $\tau^{c}=\{\,G^{c}\mid G\in \tau\,\}$ is also a topology on $X$, then $\tau^{c}$ is the dual of $\tau$.
 \end{definition}

 For the  first time, Sun and Wang \cite{31} introduced mathematical ideas for  big data search. Later, Sun \cite{38} extended some ideas of \cite{31}. Here, we restate their definitions as follows: Let $u\in U$ be a document on the Web. Then $u$ can be a Microsoft Word file in .docx or a report in PDF. Let $v$ be an attribute value. Then, $v$ may be a word such as `big', `data',` analytics', `intelligence', etc. \\
 
\begin{definition}\cite{38}
A search function, denoted as $S: V \rightarrow U$, is defined as
$S(v)=u$ if $v\in u$.
For example, if we use Google to search `analytics', denoted as $v$, then we search a file on business analytics services, denoted as $u$ including $v$.
\end{definition}
\begin{theorem} \cite{31}
The search results, with regard to semantic union, `$\vee$'  in the finite universe of big data is \par $S(v_{1} \vee v_{2} \vee v_{3} \vee v_{4} \vee...\vee v_{n})$= $S(v_{1}) \cap S(v_{2})\cap S( v_{3})\cap S( v_{4})\cap...\cap S(v_{n})$.
\end{theorem}
\begin{theorem} \cite{31}
The search results, with regard to semantic intersection, `$\wedge$'  in the finite universe of big data is \par $S(v_{1} \wedge v_{2} \wedge v_{3}\wedge  v_{4}\wedge...\wedge v_{n})$=$S(v_{1}) \cup S( v_{2})\cup S( v_{3})\cup S( v_{4})\cup...\cup S(v_{n})$
\end{theorem}

For the first time, Acharjee et al. \cite{39}, gave the definition of primal on a non-empty set. Here, the definition of primal is given below. \\
\begin{definition}\cite{39}
    Let $X$ be a non-empty set. A collection $\mathcal{P}\subseteq 2^{X}$ is called primal on $X$ if it satisfies the following conditions:\\ (i) $X\notin \mathcal{P},$\\ (ii) if $A\in \mathcal{P}$ and $B\subseteq A$, then $B\in \mathcal{P}$,\\ (iii) if $A\cap B\in\mathcal{P}$, then $A\in \mathcal{P}$ or $B\in \mathcal{P}.$
    \end{definition}

\section{ Big data search function through relation:}
In the study of mathematical ideas for big data searching, to start with the study of the neighborhood structure of words, selecting an appropriate data representation approach is vital for effectively capturing and analyzing the relationships between data points. The data representation method should reflect the inherent connections and dependencies within the data set, enabling the identification of neighbors and facilitating subsequent analysis.\\

In our day-to-day lives, we search for various words on Google, Facebook, YouTube, etc. Here, for each word, we search and we get the results in some patterns. For example, if we search words viz., `big data' and `big' in Google, then we get almost 3,42,00,00,000 and 6,30,00,00,000 results, respectively, in some particular patterns (retrieved on 07.09.2023), as shown in figures 1 and 2, respectively. Here, we can notice that all documents in the search space of `big data' are included in the search space of the word `big', i.e., $S(big$ $ data)\subseteq S(big).$ \\
\begin{figure}[H]
   \centering \includegraphics[width= 6.5 in]{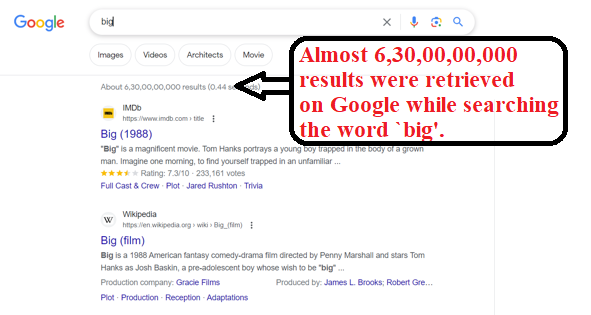 }
         \caption{ Search result for `big' on Google displaying the volume of data retrieved.}
         \label{fig: Search result for `big' on Google displaying the volume of data retrieved.}
          \end{figure}

\begin{figure}[H]
   \centering
         \includegraphics[width=6.5in]{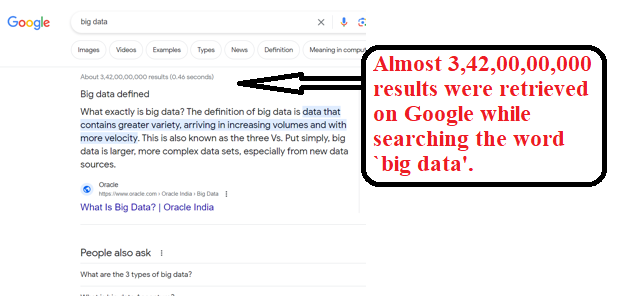 }
         \caption{Search result for `big data' on Google displaying the volume of data retrieved.}
         \label{fig:Search result for `big data' on Google displaying the volume of data retrieved.}
    \end{figure}

Such scenarios motivate the development of new mathematical concepts tailored to the challenges of big data and its search processes. We describe these concepts below:\\

     \begin{definition}
     Let $B$ be the universe of big data, $W$ be the set of words, and $S$ be a search function. Now, for any two words $x,y\in W$, we define a relation $\mathbf{R}$ on $W$ such that $x\mathbf{R}y\iff S(y)\subseteq S(x)$. We can write it as $\mathbf{R}=\{\,(x,y)\mid S(y)\subseteq S(x)$ and $ x,y\in W\,\}$. Here, $S(x)\subseteq B $ $  \forall x\in W.$
   \end{definition}  

Now, we notice that if we search `big' and `big data' in a search engine, say Google, then $S(big$ $data)\subseteq S(big)$, but $S(big)\nsubseteq S(big$ $data)$. So, there may be some ordered property between the words `big' and `big data' in the context of big data searching. Thus, we discuss the following:

     \begin{theorem}
    The relation $\mathbf{R}$ on $W$ such that, for any $x,y\in W$, $x\mathbf{R}y\iff S(y)\subseteq S(x)$ is a preorder relation.
      \end{theorem}
      \begin{proof}
          \textit{Reflexivity:} \\ Let for any $x\in W$ we have $S(x)\subseteq S(x).$  So, $(x,x)\in$ $\mathbf{R}$ $\forall x\in W$.\\ \textit{Transitivity:}\\ Let for any $x, y, z\in W$, such that $(x,y)\in \mathbf{R}$ and $(y,z)\in \mathbf{R}$. Then, by definition 4.1, we have $S(y)\subseteq S(x)$, and $S(z)\subseteq S(y)$. Then, $S(z)\subseteq S(y)\subseteq S(x)$. We can say that $S(z)\subseteq S(x)$. This implies that $(x,z)\in \mathbf{R}.$ Thus, $\mathbf{R}$ is a preorder relation.
\end{proof}
 \begin{remark}
          Since the relation $\mathbf{R}$ on $W$ such that $x\mathbf{R}y\iff S(y)\subseteq S(x)$ is preorder, so we call $(B,W,S,\mathbf{R})$ as preorder big data system (POBDS). From now onwards, in this paper $(B,W,S,\mathbf{R})$ will be known as a preorder big data system.
      \end{remark}
        \section{ Neighborhood systems induced by relation:}

Sierpiński first introduced the notion of neighborhood systems for studying Fréchet (V) spaces \cite{40}. It was developed from the idea of the geometric notion of nearness \cite{40}. Neighborhood structure analysis in big data can provide significant advantages and insights in various fields. Neighborhood structures in big data may refer to those data sets where data points are related to each other through some suitable relations.\\

From the previous section, it is clear that for any two words $x, y$ in $W$, either $S(x)\subseteq S(y)$ or $S(y)\subseteq S(x).$ For example, $S(World$ $bank)\subseteq S(World)$, but $S(World)\nsubseteq S(World$ $bank).$ Thus, in the following part, we have defined some notions of forneighborhood and afterneighborhood of words in big data search:
\begin{definition}
     In $(B,W,S,\mathbf{R}),$ forneighborhood and afterneighborhood of a word $x\in W$ are defined as $\mathbf{R}x=\{\,y\mid S(x)\subseteq S(y)\,\}$ and $x\mathbf{R}=\{\,y\mid S(y)\subseteq S(x)\,\}$ respectively. Moreover, topologies generated by $\{\,\mathbf{R}x\mid x\in X\,\}$ and $ \{\,x\mathbf{R}\mid x\in X\,\}$ as subbases are denoted as $\tau_{F}$ and $\tau_{B}$ respectively.
 \end{definition}
 \noindent

    \begin{lemma}
        In $(B,W,S,\mathbf{R}),$ let $\mathbf{R}$ be a relation on $W$ such that $x\mathbf{R}y\iff S(y)\subseteq S(x).$ For any $x,y\in W $, we have
        \begin{enumerate}[(i)]
            \item  $D=\cup_{x\in D}x\mathbf{R}$  if $D\in \tau_{B}$,
            \item $D^{/}=\cup_{x\in D^{/}}\mathbf{R}x$  if $D^{/}\in \tau_{F}.$
        \end{enumerate}
         \end{lemma}
\begin{proof} \begin{enumerate}[(i)]
    \item Let $D\in \tau_{B}$ . To show that $D=\cup_{x\in D}x\mathbf{R}.$ Let $z\in D$. Then, we have $S(z)\subseteq S(z).$ So,  $z\in z\mathbf{R}$. Hence, $z\in\cup_{x\in D} x\mathbf{R} $. Therefore, $D\subseteq \cup_{x\in D}x\mathbf{R}$.\par Conversely, let $y\in\cup_{x\in D} x\mathbf{R}$. Then, there exists $x\in D$ such that $y\in x\mathbf{R}$. So, by definition 5.1, we have $S(y)\subseteq S(x). $ Clearly, $y\in y\mathbf{R}$ since $S(y)\subseteq S(y)$. Again, if $w\in y\mathbf{R}$, then we have $S(w)\subseteq S(y)$. So, $S(w)\subseteq S(y)\subseteq S(x)$. Hence, $S(w)\subseteq S(x)$. Thus, $w\in x\mathbf{R}. $ Therefore, $y\mathbf{R}$ is the smallest open set containing $y$. So, $y\in y\mathbf{R}\subseteq D. $ Hence, $ \cup_{x\in D}x\mathbf{R}\subseteq D$. Hence, $D=\cup_{x\in D}x\mathbf{R}$.
\item The proof of this part can be obtained by using a similar process.
\end{enumerate}
\end{proof}
The topology $\tau_{F}$  facilitates forward search. Here, we start with a general word first and then start exploring broader neighborhoods of that word. For example, when a user searches for $`Big$ $Data$ $Analytics$' on Google, and if he is not satisfied with the findings, then the user will search forward by expanding the search space to $`Big$ $Data$'. So, from figure 4, we can see that $S(Big$ $Data $ $ Analytics)\subseteq S(Big $ $Data)$,  and thus $Big$ $Data\in \mathbf{R}(Big$ $Data$ $Analytics)$, i.e., forneighborhood of $Big$ $Data$ $Analytics$ contains the word $Big$ $data$ and thus $\tau_{F}$ is formed.   In contrast, the topology $\tau_{B}$ helps us identify specific or narrower contexts that are directly related to the searched word. For example, from figure 4, we can see that when a user searches for the word $`Big$' on Google, then he will get a broader view. So,  in order to be more specific, he may search for $`Big$ $ Movie$'. Thus, $S(Big $ $ Movie)\subseteq S(Big),$. Hence, we get $`Big$ $ Movie$'$\in Big \mathbf{R}$. Therefore, we can see that in the same data set, the $\tau_{F}$ helps to include generalization of context, while the $\tau_{B}$ helps to specialize the context of the searched word by narrowing down the search space. Hence, they are dual in nature. This duality serves as a tool in our method, allowing the user to choose between a broad, exploratory search or a more precise search within big data.\\
 \begin{theorem}
   In $(B,W,S,\mathbf{R}),$ let $\mathbf{R}$ be a relation on $W$ such that $x\mathbf{R}y \iff S(y)\subseteq S(x)$, for any $x,y \in W$. Then, the topologies $\tau_{F}$ and $\tau_{B}$ are dual to each other.\end{theorem}
\begin{proof}
    Let $D\in\tau_{F}$. To show that $D^{c}\in \tau_{B}$, i.e., $D^{c}=\cup_{x\in D^{c}}x\mathbf{R}.$ Let $y\in D^{c}.$ Then, $y\notin D=\cup_{x\in D}\mathbf{R}x$. Thus $\forall x\in D$, we have $y\notin \mathbf{R}x$. Hence, $S(x)\nsubseteq S(y).$ So, there exists $z\in D^{c}$ such that $S(y)\subseteq S(z)$. Therefore, $y\in\cup_{x\in D^{c}}x\mathbf{R}$ and hence, $D^{c}\subseteq\cup_{x\in D^{c}}x\mathbf{R}$.\par  Let $w\in \cup_{x\in D^{c}}x\mathbf{R}$. Then, there exists $x\in D^{c}$ such that $w\in x\mathbf{R}$. Thus, $S(w)\subseteq S(x)$. Again, as $x\in D^{c}$ implies $x\notin D=\cup_{z\in D}\mathbf{R}z$, and so for all $z\in D$ such that $S(z)\nsubseteq S(x).$ But as $S(w)\subseteq S(x)$, so we have $w\notin D.$ Thus, $w\in D^{c}$ and so, $\cup_{x\in D^{c}}x\mathbf{R}\subseteq D^{c}$. Hence, $D^{c}=\cup_{x\in D^{c}}x\mathbf{R}$. Similarly, we can prove for the other part in the case of $\tau_{B}.$
\end{proof}

If we search words like \textit{`bang bang'} on Google, then clearly the phrase \textit{`bang bang'} belongs to both the forneighborhood and afterneighborhood of the word \textit{`bang'}. Thus, if $A$ is a collection of all such words, then clearly $A\in\tau_{B}\cap \tau_{F}$. The following theorem is based on this concept:\\
\begin{theorem}
  In $(B,W,S,\mathbf{R})$, let $\mathbf{R}$ be a relation on $W$ such that $\mathbf{R}=\{\,(x,y)\mid S(y)\subseteq S(x)\,\}$. If $A\in \tau_{B}\cap \tau_{F}$, then $ A\subseteq\cup_{x,z\in A}(x\mathbf{R}\cap \mathbf{R}z).$
\end{theorem}
\begin{proof}
    Let $A\in \tau_{B}\cap \tau_{F}.$ Then $A\in \tau_{B} $ and $A \in \tau_{F}$. From {lemma 5.1},  $A\in \tau_{B}$ implies  $A=\cup_{x\in A}x\mathbf{R}$, and also, $A\in \tau_{F}$ implies  $A=\cup_{x\in A}\mathbf{R}x$. Thus, for each $y\in A$, there exist $x\in A$, and $z\in A$ such that $y\in x\mathbf{R}$ and $y\in \mathbf{R}z$. It implies $y\in x\mathbf{R}\cap \mathbf{R}z$. Thus, $y\in \cup_{x,z\in A}(x\mathbf{R}\cap \mathbf{R}z).$ So, we have $A\subseteq\cup_{x,z\in A}(x\mathbf{R}\cap \mathbf{R}z)$.
    \end{proof}
    \begin{corollary}
        In $(B,W,S,\mathbf{R})$, if  $A\in \tau_{B}\cup \tau_{F}$, then $A=\cup_{x\in A}(x\mathbf{R}\cup \mathbf{R}x).$
    \end{corollary}
    \begin{proof}
       It is given that  $A\in \tau_{B}\cup \tau_{F}$. So, from {lemma 5.1}, we have $A=\cup_{x\in A}\mathbf{R}x$ or $A=\cup_{x\in A}x\mathbf{R}$. Hence, $A=(\cup_{x\in A}\mathbf{R}x)\cup (\cup_{x\in A}x\mathbf{R} )$. So, $A=\cup_{x\in A}(\mathbf{R}x\cup x\mathbf{R}).$
    \end{proof} \par
  In big data analytics, users commonly navigate vast and heterogeneous information spaces using keyword-based searches, recommendation engines, and content exploration tools \cite{41,42}. %However, real-world queries are frequently ambiguous, lacking contextual sensitivity, or embedded with multiple semantic layers. 
 For instance, a search for ``Psychology'' might implicitly refer to more specific subdomains such as ``Cognitive Psychology,'' ``Memory,'' or ``Clinical Disorders,'' each associated with its own set of relevant documents. %In the absence of robust semantic structures, traditional search systems tend to yield flat, unorganized, and often irrelevant results.
 \\

There is a practical need to construct hierarchical or neighborhood-based structures over keywords that capture semantic containment. Specifically, if the set of documents retrieved for a keyword $y$ is a subset of those retrieved for another keyword $x$ (i.e., $S(y) \subseteq S(x)$), then $y$ can be considered semantically or contextually more specific than $x$. Identifying and modeling such relationships enable the system to:\\

\begin{enumerate}[(i)]

    \item suggest more specific or broader search keywords based on a user's intent,
  
    \item enable a zoom-in or zoom-out mechanism in interactive search interfaces,
    \item support explainable and traceable recommendation systems, where users can see how concepts are related.
\end{enumerate}
We therefore present an efficient algorithm for constructing forneighborhoods and afterneighborhoods of words.

\begin{algorithm}[H]
\caption{Algorithm for Constructing Forneighborhoods and Afterneighborhoods of Words in Big Data Search. }
\KwIn{Set of keywords $W = \{w_1, w_2, \dots, w_n\}$; Search function $S(w)$ returning search space of $w$
\KwOut{Afterneighborhoods $xR$, Forneighborhoods $Rx$ for all $x \in W$}}

\ForEach{$x \in W$}{
    Initialize $xR \gets \emptyset$\;
    Initialize $Rx \gets \emptyset$\;
    \ForEach{$y \in W$}{
        \If{$S(y) \subseteq S(x)$}{
            $xR \gets xR \cup \{y\}$ 
        }
        \If{$S(x) \subseteq S(y)$}{
            $Rx \gets Rx \cup \{y\}$ 
        }
    }
}
\Return{All $xR$, $Rx$}
\end{algorithm}

\begin{example}

To demonstrate the practical applicability of the aforementioned algorithm, we present a domain-specific scenario based on hierarchical terminology from the field of psychology. The objective is to model search relationships among conceptual keywords where theoretical inclusion is well established—allowing us to reasonably assume that  \( S(y) \subseteq S(x) \).\\

\noindent
In psychology literature searches using resources from the American Psychological Association—such as APA PsycINFO—researchers explore topics ranging from broad psychological domains to specific theories or disorders. Constructing a neighborhood structure based on domain knowledge enables hierarchical filtering, precise query suggestions, and semantically meaningful document grouping.\\

\noindent
For this purpose,  we consider the following terms commonly used in psychology: W = \{Psychology, Cognitive Psychology, Memory, Working Memory, Short-Term Memory, Long-Term Memory, Clinical Psychology, Depression, Anxiety, Cognitive Behavioral Therapy\}. Based on established domain knowledge, we assume the following theoretical subset relationships among their corresponding search spaces:\\
\begin{enumerate} [(i)]  
\item $S(Cognitive $ $  Psychology) \subseteq S(Psychology)$
    \item S(Memory) $\subseteq$ S(Cognitive $ $ Psychology)
    \item $S(Working $ $   Memory)\subseteq S(Memory)$.
    \item S(Short-Term $ $  Memory)$\subseteq$ S(Memory)
    \item S(Long-Term $ $  Memory) $\subseteq$ S(Memory)
    \item $S(Clinical $ $  Psychology) \subseteq S(Psychology)$
    \item $S(Depression)\subseteq S(Clinical $ $   Psychology)$,
    \item $S(Depression)\subseteq S(Clinical $ $   Psychology)$,
    \item $S(Depression)\subseteq S(Clinical $ $   Psychology)$,
    \item $S(Depression)\subseteq S(Clinical $ $   Psychology)$,
    \item $S(Anxiety) \subseteq S(Clinical $ $   Psychology)$
    \item $S(Cognitive $ $   Behavioral  $ $  Therapy) \subseteq S(Depression) \cap S(Anxiety)$
\end{enumerate}
 Now, using Algorithm 1, we compute the following tables of   \( xR \)and  \( Rx\):\\
\begin{table}[H]
\centering
\fontsize{9pt}{11pt}\selectfont
\begin{tabular}{|l|p{10cm}|}
\hline
\textbf{Terms (x)} & \textbf{Afterneighborhood ($xR$)} \\
\hline
Psychology & \{Psychology, Cognitive Psychology, Clinical Psychology, Memory, Working Memory, Short-Term Memory, Long-Term Memory, Depression, Anxiety\} \\
Cognitive Psychology & \{Cognitive Psychology, Memory, Working Memory, Short-Term Memory, Long-Term Memory\} \\
Memory & \{Memory, Working Memory, Short-Term Memory, Long-Term Memory\} \\
Working Memory & \{Working Memory\} \\
Short-Term Memory & \{Short-Term Memory\} \\
Long-Term Memory & \{Long-Term Memory\} \\
Clinical Psychology & \{Clinical Psychology, Depression, Anxiety\} \\
Depression & \{Depression\} \\
Anxiety & \{Anxiety\} \\

\hline
\end{tabular}
\caption{Afterneighborhoods $xR$ for search terms in psychology.}
\end{table}

\medskip

\begin{table}[H]
\centering
\fontsize{9pt}{11pt}\selectfont
\begin{tabular}{|l|p{10cm}|}
\hline
\textbf{Terms (x)} & \textbf{Forneighborhoods ($Rx$)} \\
\hline
CBT & \{Psychology, Clinical Psychology, Depression, Anxiety\} \\
Working Memory & \{Psychology, Cognitive Psychology, Memory, Working Memory\} \\
Short-Term Memory & \{Psychology, Cognitive Psychology, Memory, Short-Term Memory\} \\
Long-Term Memory & \{Psychology, Cognitive Psychology, Memory, Long-Term Memory\} \\
Memory & \{Psychology, Cognitive Psychology, Memory\} \\
Cognitive Psychology & \{Psychology, Cognitive Psychology\} \\
Clinical Psychology & \{Psychology, Clinical Psychology\} \\
Depression & \{Psychology, Clinical Psychology, Depression\} \\
Anxiety & \{Psychology, Clinical Psychology, Anxiety\} \\
Psychology & \{Psychology\} \\
\hline
\end{tabular}
\caption{Forneighborhoods $Rx$ for search terms in psychology.}
\end{table}
\end{example}
\medskip
\subsubsection{$m$- steps relation of $\mathbf{R}$ in $(B, W, S,\mathbf{R})$ :}
     Let us choose three words `Big', `Big Data', and `Big Data Analytics'. Practically, it is important to note that, \textit{S(Big Data)} $\subseteq$ \textit{S(Big)} and \textit{S(Big Data Analytics)} $\subseteq$ \textit{S(Big Data)}. Here, we can have \textit{S(Big Data Analytics)$\subseteq$ S(Big Data) $\subseteq$ S(Big)}. So, \textit{S(Big Data Analytics) $\subseteq$ \textit{S(Big)}}. Mathematically, if we consider $x=$\textit{ Big}, $y=$ \textit{Big Data}, and $z= $ \textit{Big Data Analytics}, then $S(y)\subseteq S(x)$ and $S(z)\subseteq S(y)$. Hence, $S(z)\subseteq S(x)$. So, in
$(B, W, S, R)$, if $x \mathbf{R} y$ and $y \mathbf{R} z$, then $x \mathbf{R} z$. From this, we define $\mathbf{R}^{2}=\{\,(x,z)\mid $ there exists $ y\in W $ such that $ x \mathbf{R} y , y \mathbf{R} z\,\}$. \par 
      Above notion motivates us to define an $m$- steps relation of $\mathbf{R}$ in $(B, W, S, \mathbf{R})$ as follows: \par 
    \begin{definition} 
 In $(B, W, S, \mathbf{R})$, we can define an $m$-steps relation as $\mathbf{R}^{m}=\{\,(x,z)\mid $ there exist $y_{1}, y_{2},...,y_{m-1} \in W$ such that $ S(z)\subseteq S(y_{m-1})\subseteq...\subseteq S(y_{3})\subseteq S(y_{2})\subseteq S(y_{1})\subseteq S(x)\,\}$.
\end{definition}

Practically, when users search for a query in big data, they may not obtain the desired results on their first attempt. Instead, they often refine or adjust their search queries iteratively. The $m$-steps relation, denoted as $\mathbf{R}^{m}$, captures this natural behavior by representing a chain of increasingly refined or generalized search terms. Each step corresponds to a search that yields results that are either more contextually specific or broader. The following example illustrates the $m$-steps relation $\mathbf{R}^{m}$:\\

\begin{example}
 Let us consider a scenario where the user is searching for health-related data. For the user's first search, $m=1$, the user uses the keyword ``\textit{Health}''. It is a very broad term that gives information including diet, mental health, public health, etc. If the user cannot find the desired result, then in the second search, $m=2$, the user refines the search query to, say, ``\textit{Public health}'', which is a more specific term than the first, focused on information related to epidemiology, government policies, etc. If the user is satisfied with the result, then it stops here; otherwise, the user opts for the third search, $m=3$, say, ``\textit{Public health policy in India}''. This sequence represents a 3-steps relation $\mathbf{R}^{3}$ in our framework, where each term's search space is a subset of the previous one: \textit{S(Public health policy in India) $\subseteq$ S(Public health) $\subseteq$ S(Health).} This way we can think of the $m-$steps relation, and this nested process aligns with the user's exploration.\\
    
\end{example}
In the following result, we establish the consistency of the relation $\mathbf{R}^{m}$. Practically, when a user searches for a word using a search engine, it is evident that if the initial results are not highly relevant, the user will continue refining the query for a finite number of steps—say, $m$—until relevant information is found. In our framework, this consistency can be demonstrated by showing that the relation $\mathbf{R}^{m}$ is serial. A serial relation implies that for every word $x$, there exists a word $y$ such that $x$ is related to $y$. In practical terms, this means that for each word $x$ being searched, even after $m$ refinement steps, there must exist a word $y$ in the dataset such that $S(y) \subseteq S(x)$.\\
 \begin{theorem}
          In $(B, W, S, \mathbf{R})$, if $\mathbf{R}$ is serial, then $\mathbf{R}^{m}$ is a serial relation for all $m\geq 1$.
          \end{theorem}
          \begin{proof}
                 It is given that $\mathbf{R}$ is serial. So, for each $x\in W$, there exists $y\in W$ such that $x\mathbf{R
}y$ and it implies that $S(y)\subseteq S(x).$ We aim to show that $\mathbf{R}^{m}$ is a serial relation. Since $\mathbf{R}$ is serial for each $x\in W$,  we have $y_{1}\in W$ such that $x\mathbf{R}y_{1}$. Similarly, since $y_{1}\in W$, there exists $y_{2}\in W$ such that $y_{1}\mathbf{R}y_{2}$. In a similar manner, for each $y_{i}\in W$, there exists $y_{i+1}\in W$ such that $y_{i} \mathbf{R} y_{i+1}$, where $i=1,2,...,m-1$. We define $y_{m}=y.$ Now, $x\mathbf{R}y_{1}$ implies $S(y_{1})\subseteq S(x).$ Similarly, continuing this process,  we obtain $S(y_{i+1})\subseteq S(y_{i})$. Then, $S(y)\subseteq S(y_{m-1})\subseteq...\subseteq S(y_{2})\subseteq S(y_{1})\subseteq S(x)$. Thus, for each $x$, there exists $y$ in $W$ such that $x\mathbf{R}^{m}y$. Hence, by definition 3.2,  $\mathbf{R}^{m}$ is a serial relation.
 \end{proof}
 \begin{theorem}
              In $(B, W, S, \mathbf{R})$, $\mathbf{R}^{m}$ is a preorder relation.
          \end{theorem}
    \begin{proof}
              \textit{Reflexivity:}
         This is obvious since, for any $x\in W$, $S(x)\subseteq S(x).$ It implies that $S(x)\subseteq S(x)\subseteq...\subseteq S(x)\subseteq S(x)\subseteq S(x)$(up to $m$-times). So, for each $x\in W$, $x \mathbf{R}^{m} x$. Hence, $\mathbf{R}^{m}$ is a reflexive relation.\par
     \textit{Transitivity:}
Here, $\mathbf{R}=\{\,(x,y)\mid x,y \in W$ and $S(y)\subseteq S(x)\,\}$. To show that, for any $x,y,z\in W, x \mathbf{R}^{m} z$ and $z\mathbf{R}^{m} y$ implies $x \mathbf{R}^{m} y.$ If $ x \mathbf{R}^{m} z $, then there exists $y_{i}
, i=1,2,...,m-1$ such that $S(z)\subseteq S(y_{m-1})\subseteq...\subseteq S(y_{2})\subseteq S(y_{1}) \subseteq S(x).$ Again if $z\mathbf{R}^{m}y$, then there exist $y^{/}_{i}, i=1,2,...,m-1$ such that $S(y)\subseteq S(y^{/}_{m-1})\subseteq...\subseteq S(y^{/}_{2})\subseteq S(y^{/}_{1}) \subseteq S(z)$. Thus, we have $S(y)\subseteq S(y^{/}_{m-1})\subseteq...\subseteq S(y^{/}_{2}
)\subseteq S(y^{/}_{1}) \subseteq S(z)\subseteq S(y_{m-
1})\subseteq...\subseteq S(y_{2})\subseteq S(y_{1}) \subseteq S(x)$. So, we have $S(y)\subseteq S(y_{m-1})\subseteq...\subseteq S(y_{2})\subseteq S(y_{1}) \subseteq S(x).$ This implies that $x\mathbf{R}^{m}y$. Hence, $\mathbf{R}^{m}$ is a transitive relation. Thus, $\mathbf{R}^{m}$ is a preorder relation.  \end{proof}
\begin{remark}
    Since $\mathbf{R}^{m}$ is a preorder in $(B, W, S, \mathbf{R})$, hence the system $(B, W, S, \mathbf{R}^{m})$ is also a POBDS.
\end{remark} 
Now, we define forneighborhood and afterneighbourhood for 
$R^{m}$ as given below:
               \begin{definition}
                   In $(B,W,S,\mathbf{R}^{m})$, forneighborhood of $x$ for $\mathbf{R}^{m}$ is $\mathbf{R}^{m}_{F}(x)=\{\,y\in W \mid$ there exist $ y_{1},y_{2},...,y_{m-1} $ such that $S(x)\subseteq S(y_{1}) \subseteq S(y_{2})\subseteq...\subseteq S(y_{m-1})\subseteq S( y)\,\}$ and afterneighborhood of $x$ is $\mathbf{R}^{m}_{A}(x)=\{\,y\in W \mid$ there exist $y^{/}_{1},y^{/}_{2},...,y^{/}_{m-1}$ such that $S(y)\subseteq S(y^{/}_{m-1}) \subseteq S(y^{/}_{m-
2})\subseteq...\subseteq S(y^{/}_{2})\subseteq S( y^{/}_{1})\subseteq S(x)\,\}$
               \end{definition}
               Here, we present key results concerning the hierarchical structure of the $m$-steps relation. As illustrated in figure 4, we observe that appending terms  to the word `$Big$' at each step at each step causes the associated search space to contract. Specifically, as the number of search steps $m$ increases, the afterneighborhood of $`Big$' refines the search space, making it more specific and relevant. In contrast, the forneighborhood of a word expand the search space, rendering it broader and more contextually inclusive. The following definitions and results formalize these properties within the framework of our big data search approach:\\
               
               \begin{definition}
                   Let $E$ and $F$ be two subsets of $W$. Then, 
                   \begin{enumerate}
                    
              \item a relation `$\preceq$' on subsets of $W$ such that $E\preceq F$ as   if for each $y\in F$, there exists $z\in E$ such that $S(y)\subseteq S(z)$,
                       \item a relation $`\succeq'$ on  subsets of $W$ such that $E\succeq F$ as for each $y\in F$, there exists $z\in E$ such that $S(z)\subseteq S(y).$ 
                        \end{enumerate}
               \end{definition}
\begin{theorem}
          In $(B,W,S,\mathbf{R}^{m})$, the neighborhood system $\{\,\mathbf{R}^{m}_{A}(x)\mid m\geq 1\,\}$ satisfies the condition $\mathbf{R}^{m}_{A}(x)\preceq \mathbf{R}^{m+1}_{A}(x),\forall m\in \mathbf{N}.$
      \end{theorem}
      \begin{proof}
          Let $y\in \mathbf{R}^{m+1}_{A}
(x)$, $\forall m\in N$. Then, there exist $y_{1},y_{2},...,y_{m}$ in $W$ such that $S(y)\subseteq S(y_{m}) \subseteq S(y_{m-1})\subseteq...\subseteq S(y_{2})\subseteq S(y_{1})\subseteq S(x).$ This implies that $y_{i}\in \mathbf{R}^{i}_{A}(x),\forall i=1,2,..,m.$ Thus, for each $y\in \mathbf{R}^{m+1}_{A}(x)$, there exists $y_{m}\in \mathbf{R}^{m}_{A}(x)$ such that $S(y)\subseteq S(y_{m}).$ Hence, by definition 5.4, we have $\mathbf{R}^{m}_{A}(x)\preceq \mathbf{R}^{m+1}_{A}(x),\forall m\in \mathbf{N}$.
     \end{proof}

     \begin{theorem}

          In $(B,W,S,\mathbf{R}^{m})$, the neighborhood system $\{\,\mathbf{R}^{m}_{F}(x)\mid m\geq 1\,\}$ satisfies the condition $\mathbf{R}^{m}_{F}(x)\succeq \mathbf{R}^{m+1}_{F}(x),\forall m\in \mathbf{N}.$
    \end{theorem}

     \begin{proof}
       Let $y\in \mathbf{R}^{m+1}_{F}(x), \forall m\in N. $ Then, there exist $y_{1},y_{2},...,y_{m} $ in $ W $ such that $S(x)\subseteq S(y_{1}) \subseteq S(y_{2})\subseteq...\subseteq S(y_{m-1})\subseteq S(y_{m})\subseteq S(y).$ This implies that $y_{i}\in \mathbf{R}^{i}_{F}(x), \forall i=1,2,...,m.$ Thus, for each $y\in \mathbf{R}^{m+1}_{F}(x)$, there exists $y_{m}\in \mathbf{R}^{m}_{F}(x)$ such that $S(y_{m})\subseteq S(y).$ Hence, by  definition 5.4, we have  $\mathbf{R}^{m}_{F}(x)\succeq \mathbf{R}^{m+1}_{F}(x),\forall m\in \mathbf{N}.$
     \end{proof}
      The above result helps us to construct a chain of neighborhoods of $x\in W$ in $(B,W,S,\mathbf{R}^{m})$ such as $\mathbf{R}^{1}_{A}(x)\preceq \mathbf{R}^{2}_{A}(x)\preceq \mathbf{R}^{3}_{A}(x)\preceq...\preceq \mathbf{R}^{m}_{A}(x)\preceq...$. It leads to the concept of topology generated by this chain in $(B,W, S,\mathbf{R}^{m})$ in the following section:
\begin{definition}  In $(B,W,S,\mathbf{R}^{m})$, the topologies $\mathcal{N}_{F}$ and $\mathcal{N}_{A}$ can be defined as the topologies generated by the collections $\{\,\mathbf{R}^{m}_{F}(x)\mid x\in W\,\}$ and $\{\,\mathbf{R}^{m}_{A}(x)\mid x\in W\,\}$ as subbasis respectively.
\end{definition}
\begin{theorem}
    In $(B,W,S,\mathbf{R}^{m})$, the following results hold: \begin{enumerate}[(i)]
        \item  $D=\cup_{x\in D}\mathbf{R}^{m}_{F}(x),$ for any $D\in \mathcal{N_{F}}$,
        \item $E=\cup_{x\in E}\mathbf{R}^{m}_{A}(x),$ for any $E\in \mathcal{N_{A}}$.
    \end{enumerate}
    \end{theorem}
\begin{proof}
     Here, $\mathbf{R}^{m}$ is a preorder relation, and according to definition 5.2, it is clear that $x \mathbf{R}^{m} y \iff $ there exist $y_{1},y_{2},...,y_{m-1}$ in $W$ such that $ S(y)\subseteq S(y_{m-1})\subseteq...\subseteq S(y_{3})\subseteq S(y_{2})\subseteq S(y_{1})\subseteq S(x)$. So, we can say that $S(y)\subseteq S(x).$ Hence, by lemma 5.1, it is clear that results (i) and (ii) hold. 
\end{proof}
\noindent

The relationship between words in data searching is rarely uniform. At times, words are strongly related, while at other times, they share a weak or partial connection. For example, the words `$Big$' and `$Big$ $data$' are strongly related by our relation $\mathbf{R}$ as $S(Big$ $data)\subseteq S(Big)$, but words like `$Big$ and `$Data$', they are weakly related because their respective search spaces have some overlap only. Thus, a weighted approach for the study of the neighborhood structure of words becomes vital so that we may answer how much these types of words are related to each other.  Let us introduce a mathematical formula for the weight of relation $\mathbf{R}$, which measures the proportion of overlaps of search spaces. Let $x$ and $y$ be any two words; then the weight of $\mathbf{R}$ is\\ $${W}_{\mathbf{R}}(x,y)=\frac{|S(x)\cap S(y)|}{|S(x)|}$$.\\ It is easy to check that $0\leq {W}_{\mathbf{R}}(x,y)\leq 1$.  Clearly, this formula measures the extent to which $S(y)$ is included in $S(x)$. With this formula, we can further define the weighted forneighborhood and the weighted afterneighborhood of words in big data search. Let $x$ and $y$ be any two words and $\alpha$ predefined by the user. Then, the weighted afterneighborhood of $x$ can be defined as $$x\mathbf{R_{\alpha}}=\{\,y\mid W_{\mathbf{R}}(x,y)>\alpha\,\} $$\\ which contains all the words $y$ that depend on $x$, and $W_{\mathbf{R}}(x,y)$ measures the extent to which $S(y)$ is included in $S(x). $ For example, if $x=Big$, then $Big\mathbf{R_{\alpha}}$ contains the words like $Big$ $data, Big$ $movie$, with high weights, as their search spaces significantly overlap or subset with $S(Big). $ \\ Similarly, weighted forneighborhood of $x$ can be defined as: $$\mathbf{R_{\alpha}}x=\{\,y\mid W_{\mathbf{R}}(y,x)>\alpha\,\} $$ contains all the words $y$ on which $x$ depends, and $W_{\mathbf{R}}(y,x)$ measures the extent to which $S(x)$ is included in $S(y)$. For example, if $x= Big$ $data$, then the weighted forneighborhood of $Big$ $data$ includes broader terms like $`Big$' or $`data$' depending on the weightage.\\

In the previous section, we discussed various notions and results related to search results of a single keyword. Now, we discuss ideas related to the search space of a set of keywords, i.e., for $A\subseteq W$ in $(B, W, S, \mathbf{R})$, we define the search space of $A$ as $S(A)=\{\, S(x)\mid x\in A\,\}$.
 \begin{example}
    Let $A$ be the collection of all keywords related to a customer's purchasing behavior for groceries at Walmart. So, $S(A)$ is the collection of all search spaces for each grocery item, which may include factors such as the frequency of purchases, preferences for items across different age groups, and other related attributes.

\end{example}
    \begin{lemma}
       In $(B,W,S,\mathbf{R})$, let $A$, $B$ be two subsets of $W$ such that $A\subseteq B$.Then, $S(A)\subseteq S(B).$
        \end{lemma}
        \begin{proof}
Let $S(x)\in S(A)$ be any member. Since $A\subseteq B$, thus $x\in B$. So, $S(x)\in S(B)$. Hence, $S(A)\subseteq S(B).$  
        \end{proof}
        \begin{lemma}
    In $(B,W,S,\mathbf{R})$, let $A$ be any subset of $W$, then $S(A^{c})=S(A)^{c}.$
\end{lemma}
\begin{proof}
    Let $S(x)\in S(A^{c})$ be any element. Then, $x\in A^{c}$ if and only if $x\notin A$. So, $S(x) \notin S(A)$ if and only if $S(x)\in S(A)
^{c}$. Hence, $S(A^{c})=S(A)^{c}.$
\end{proof}
       \begin{theorem}
        Let $A$, $B$ be any two subsets of $W$. Then, $S(A\cup B)=S(A)\cup S(B).$
             \end{theorem}
        \begin{proof}
            We have $A\subseteq A\cup B $. So, by lemma 5.2, $S(A)\subseteq S(A\cup B)$. Similarly, we have $S(B)\subseteq S(A\cup B).$  Thus, 
$S(A)\cup S(B) \subseteq S(A\cup B).$ Again, let $S(x)\in S(A\cup B)$ be any element. Then, $x\in A\cup B$. So, $x\in A$ or $x\in B$. Thus, $S(x)\in S(A)$ or $S(x)\in S(B).$ Thus, $S(x)\in S(A)\cup S(B)$. Hence, $S(A\cup B)\subseteq S(A)\cup S(B)$. So, we get $S(A\cup B)=S(A)\cup S(B).$\par 
        \end{proof}   
\begin{theorem}
    Let $A$, $B$ be any two subsets of $W$. Then $S(A\cap B)=S(A)\cap S(B).$
\end{theorem}
\begin{proof}
    We know that $A\cap B\subseteq A$ and $A\cap B\subseteq B$. So, by lemma 5.2, we have $S(A\cap B)\subseteq S(A)$ and $S(A\cap B)\subseteq S(B)$. Thus, $S(A\cap B)\subseteq S(A)\cap S(B).$ Again, let $S(x)\in S(A)\cap S(B) $. Then, $S(x)\in S(A)$ and $S(x)\in S(B)$. So, by the definitions of search space of $A, B$, we have $x\in A$ and $x\in B$. Thus, $x\in A\cap B$. It implies $S(x)\in S(A\cap B)$. So, $S(A)\cap S(B)\subseteq S(A\cap B).$ Hence, $S(A\cap B)=S(A)\cap S(B).$ \par 
\end{proof}
\begin{remark}
    We can extend theorems 5.8 and 5.9 for an arbitrary family of subsets $\{\, A_{\alpha}\mid \alpha\in \Delta,  A_{\alpha}\in W\,\}$, where $\Delta$ is an index set, as $S(\cup_{\alpha\in \Delta}A_{\alpha})=\cup_{\alpha\in \Delta}S(A_{\alpha})$ and $S(\cap_{\alpha\in \Delta}A_{\alpha})=\cap_{\alpha\in \Delta}S(A_{\alpha})$
\end{remark}

\begin{corollary}
In $(B,W,S,\mathbf{R})$, let $A$ and $B$ be any two subsets of $W$. Then, the following results hold: \begin{enumerate}[(i)]
    \item $S((A\cup B)^{c})=S(A)^{c}\cap S(B)^{c},$
\item $S((A\cap B)^{c})=S(A)^{c}\cup S(B)^{c}.$
\end{enumerate}    
\end{corollary}  
        \begin{proof}
        From De Morgan's laws, we have $(A\cup B)^{c}=A^{c}\cap B^{c}$ and $(A\cap B)^{c}=A^{c}\cup B^{c}$.\par So, we have \begin{enumerate}[(i)]
            \item $S((A\cup B)^{c})=S(A^{c}\cap B^{c})=S(A)^{c}\cap S(B)^{c}$,
            \item  $S((A\cap B)^{c})=S(A^{c}\cup B^{c})=S(A)^{c}\cup S(B)^{c}$.
        \end{enumerate} 
    \end{proof}
\section{Graph-based approaches for analyzing neighborhood structure of large data set :} 
This section discusses mathematical ideas related to big data searching via graphs. A graph G is a tuple $(V, E)$ that consists of a finite set V of
vertices and a finite set E of edges; each edge is the representation of a pair of vertices\cite{43}.  In $(B, W, S, \mathbf{R})$, we can discuss a graph structure as follows: let $V\subseteq W $ be a set of nodes and $E=\{\,(x,y)\in V\times V\mid S(y)\subseteq S(x)\,\}$ be the set of edges, i.e., if $e$ is an edge between $x$ and $y$ in V, then it can be defined as the ordered pair $(x, y)$ such that $S(y)\subseteq S(x).$ 

   \begin{example}
 If we search for the word `Space', then among the various results, let us choose choose $V=\{\,$ Space.com, Outer Space, Space News\,\}. Similarly, we can choose $E$=\{\,(Space, Space.com), (Space, Outer Space), (Space, Space News)\,\}
   \end{example} 
    \begin{figure}[H]
     \centering
    \includegraphics[width=5in]{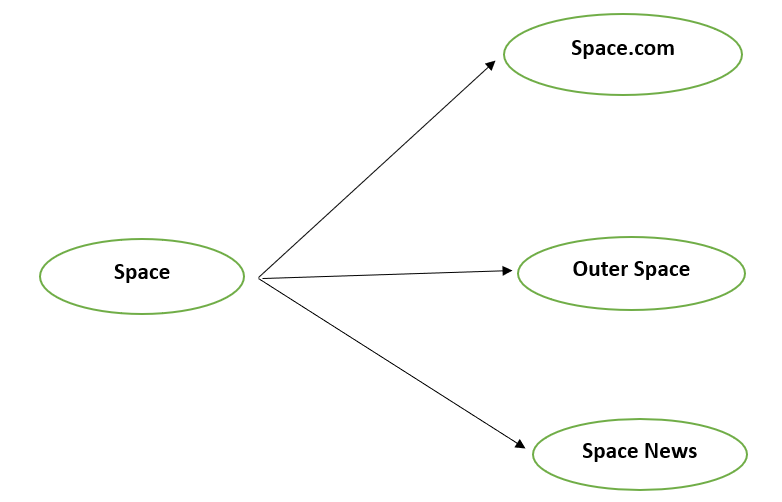}
         \caption{Graph of a part of search result for the word `Space' in Google.}
         \label{fig:Graph of a part of search result for the word `Space' in Google.}
     \end{figure}
 In some cases, as shown in example 6.2, two nodes $x,y\in V$ are not connected directly. There exists a node $y_{1}$ such that $S(y)\subseteq S(y_{1})\subseteq S(x). $ Then, edges are from $x$ to $y_{1}$ and then $y_{1}$ to $y$. This type of graph can be called a 2-steps graph. It is shown in figure 4. Similarly, we may have an \textit{m}-steps graph structure in $(B,W,S,\mathbf{R}^{m})$.
 \begin{example}
  If we search `Big' in Google, we get 17,750,00,000 results (retrieved on 25.08.23). Among them, we consider a finite set $V$=\{\, Big, Big Data, Big Movie, Big Architecture, Big Data Analytics, Big Movie Review, Big Architecture Project\,\} as a set of nodes. Since \textit{S(Big Data)} $ \subseteq $ 
\textit{S(Big), S(Big Movie)} $\subseteq$ 
 \textit{S(Big), S(Big Architecture)}  $\subseteq $ \textit{S(Big), and S(Big Data Analytics)} $\subseteq$ \textit{S(Big Data), S(Big Movie Review)} $\subseteq$ \textit{S(Big Movie), S(Big Architecture Project)} $\subseteq$ \textit{S(Big Architecture)}, so the edge set is $E$=\{\ (Big, Big Data), (Big, Big Movie), (Big, Big Architecture), (Big Data, Big Data Analytics), (Big Architecture, Big Architecture Project), (Big Movie Review, Big Movie)\,\}.\\
\end{example} 

\begin{figure}[H]
     \centering
    \includegraphics[width=5.5in]{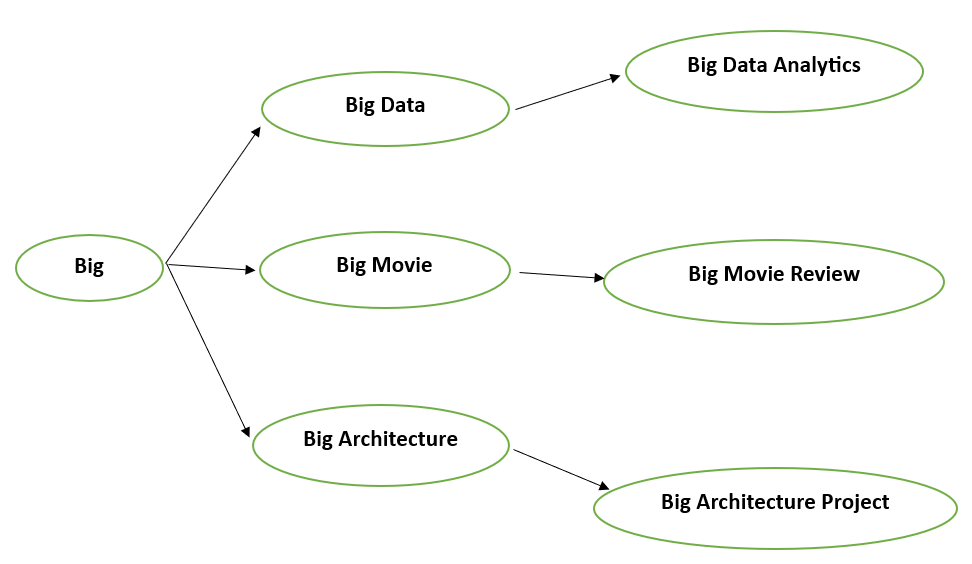}
         \caption{2-steps graph of a part of search result for the word `Big' in Google.}
         \label{fig:2-step graph for Search result for "big" in Google}
     \end{figure}
    In figure 4, it can be observed that all nodes except \textit{`Big'} can be traversed from \textit{`Big'}, but there is no node that can be traversed back to the node \textit{`Big'}. A similar case can be observed in figure 3.

\begin{definition}
    In a graph $G$ = $(V, E)$, if there is a node $x$ such that each node in $V$ can be traversed from the node $x$ but there does not exist any node that can be traversed back to the node $x$; then we call the node $x$ as the atom of the graph, and the graph itself will be coined as a data- directed graph (DDG). Here, we denote the data-directed graph as $ G^{'} = (V_{x}(G), E(G))$.
\end{definition}
\begin{example}
 Graphs in figures 3 and 4 are data-directed graphs, where the nodes \textit{`Space'} and \textit{`Big'} are atoms, respectively.
\end{example}

 If we search for $x\in W$ on Google or any other search engine, then we assume that there exist $y_{1},y_{2},y_{3},...,y_{m}$ in $W$, such that $S(y_{i})\subseteq S(x),\forall i=1,2....,m$. Thus, there is an edge from $x$ to each $y_{i}.$ Later, if we search each $y_{i}$ again, for each $i$, then there are $y_{ip}\in W$, $\forall p=1,2,...,a$ such that $S(y_{ip})\subseteq S(y_{i})\subseteq S(x).$ Thus, there are edges from each $y_{i}$ to $y_{ip}.$ By repeating the same process, we get a graph structure, and it is shown in figure 5.\\

From the available concepts in graph theory, it is evident that a directed graph cannot have a loop at any node. However, in practice, if we search for the word `Big' among 2,394,000,000 results (retrieved on 04.09.2023), we also get information about `Big'. This is shown in figure 6. Thus, $S(Big) \subseteq S(Big)$. In a data-directed graph $G{'}=(V_{x}(G), E(G)$), we have $S(x)\subseteq S(x)$, which implies that there is a loop at node $x$. Also, Imrich and Petrin \cite{44} gave the idea of a directed graph with loops. So, we propose some results related to this idea below.\\

\begin{figure}[H]
     \includegraphics[width=5.9in]{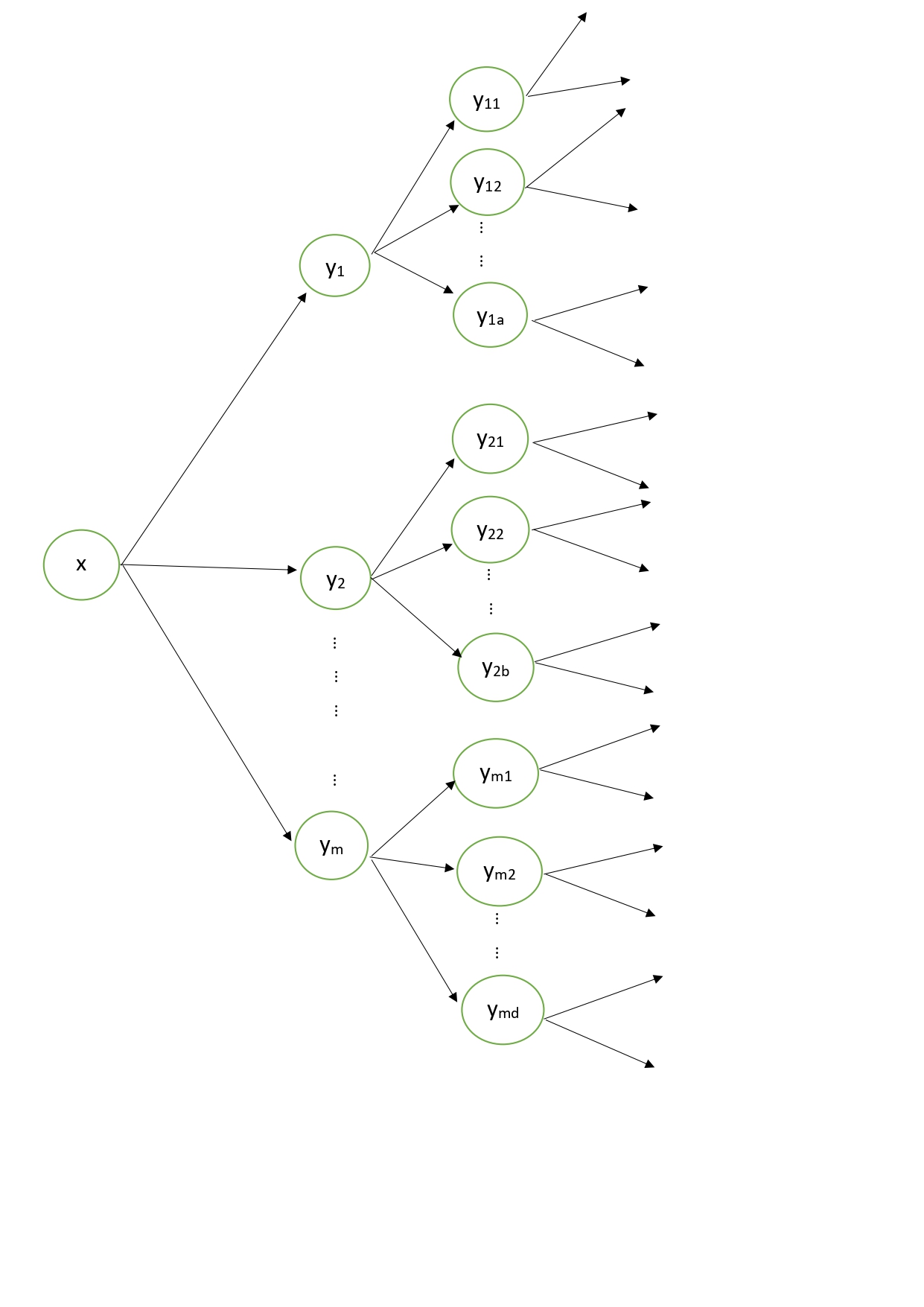}
    \caption{Data-directed graph with atom at $x$.}
    \label{fig:Data directed graph with atom at $x$. }  
     \end{figure}
 
 \begin{figure}[H]
     \centering
     \includegraphics[width= 6.5in]{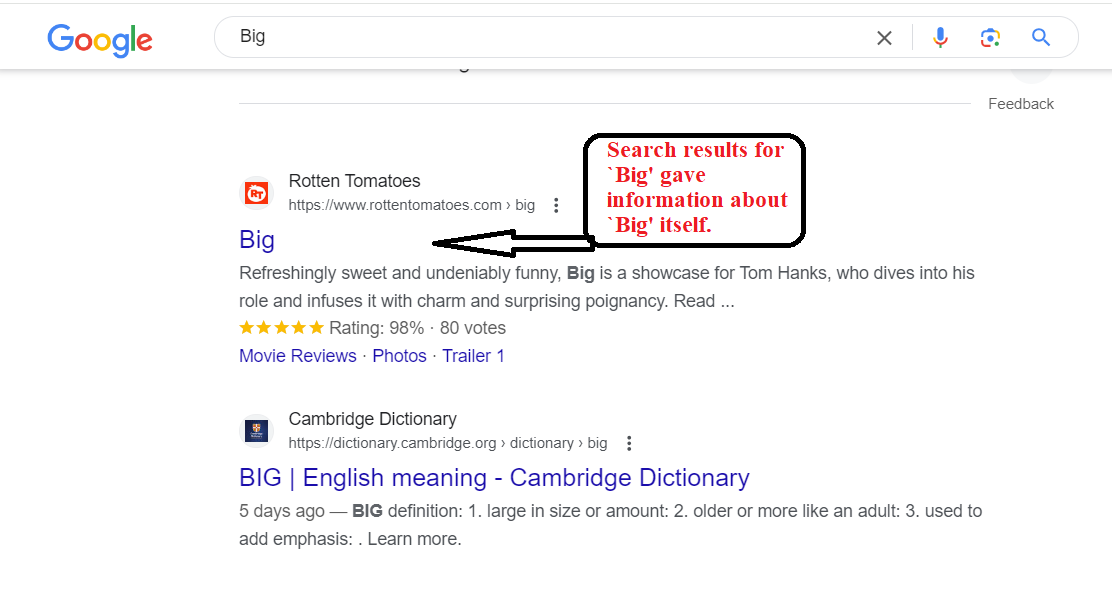}
     \caption{Search result of `Big' displaying the information for `Big' itself.}.
     \label{fig:bg2.png}
 \end{figure}
 \begin{definition}
      A directed graph is called a loop-directed graph if there are loops in some nodes.  
 \end{definition}
 \begin{proposition}
      The graph $G^{'} = (V_{x}(G), E(G))$ is a loop-directed graph.
 \end{proposition}
       \begin{proof}
 We know that a graph $G^{'} = (V_{x}(G), E(G))$ is a directed graph if and only if for any two nodes $z,y\in V_{x}(G)$, either $(z,y)\in E(G)$ or $(y,z)\in E(G).$ In $(B, W,S, \mathbf{R}^{m})$, we proved earlier that either $S(z)\subseteq S(y)$ or $S(y)\subseteq S(z)$. Again, for some $z\in V_{x}(G),$ we have $ S(z)\subseteq S(z),$ thus there exist loops in some vertex also. Hence, $G^{'} = (V_{x}(G), E(G))$ is a loop-directed graph.
\end{proof}
\begin{remark}
   The graph $G^{'} = (V_{x}(G), E(G))$ is not a tree as there may be some cycles, since for any $y\in V_{x}(G)$ there exist $z_{1},z_{2},..,z_{k-1}$ in $W$ such that $S(y)\subseteq S(z_{k-1})\subseteq S(z_{k-2})\subseteq...\subseteq S(z_{2})\subseteq S(z_{1})\subseteq S(x).$ That is $x\mathbf{R}^{k}y$, but we proved earlier that $\mathbf{R}^{k}$ is reflexive, so $x\mathbf{R}^{k}x.$ Hence, there is a cycle $xe_{1}z_{1}e_{2}z_{2}e_{3}z_{3}...e_{k-1}z_{k-1}e_{k}x$.
\end{remark}
  
\noindent
In $G^{'} = (V_{x}(G), E(G))$, if $(z,y)\in E(G)$, then $y$ is a neighbor of $z$, and the collection of such neighbors is called the neighborhood of $z$. In \cite{40}, Yao gave an idea of the distance function between two nodes in a graph. We discuss the notion of distance between nodes in DDG. \\

\begin{definition}
Let $d: V_{x}\times V_{x}\rightarrow N$ such that for any two nodes $z,y\in V$, $d(z,y)=k \iff$ there exist $z_{1},z_{2},...,z_{k-1},\in W$ such that $S(y)\subseteq S(z_{k-1})\subseteq S(z_{k-2})\subseteq...\subseteq S(z_{1})\subseteq S(z).$ 
\end{definition}

\section{Anomaly detection in big data searching :} If we search the word `pet' on Google, then we get almost 6,94,00,00,000 results (retrieved on 01.10.2023). Here, we observe that almost all the data related to pets provide information about animals or are related to animals in some way. However, among these results, there are a few that are unrelated, such as 'Polyethylene Terephthalate' (see figure 7) or 'Positron Emission Tomography scan' (see figure 8), which completely deviate from the information provided by the majority of the data. This type of data may be considered an anomaly for someone seeking information on pets, animals, and related topics. Thus, we study anomaly detection in big data search using the Jaccard similarity coefficient. Let $A$ and $B$ be any two sets. Then, the Jaccard similarity coefficient \cite{45} is defined as $J(A,B)=\frac{|A\cap B|}{|A\cup B|}$. It is evident from \cite{45} that numerically, $ 0\leq J(A, B)\leq1 .$ \\

\textbf{Step 1:}
Let us search a keyword $x$ on Google or any other search engine, and we obtain the search space of $x$ to be $S(x)=\{\,D_{1},D_{2},...,D_{m}\,\}$. Let $list_{i}, list_{j}$ be two sets of words in the data $D_{i}, D_{j}$ respectively. Then, find the Jaccard similarity coefficients of $list_{i}$ and $list_{j}$.

\textbf{Step 2:}
In this step, we consider a threshold value, say $\delta$, where 0$<\delta <$ 1, and then create a list, say $W_{0}$, containing reference keywords. Reference keywords are those that a user searches for in a search engine to obtain their desired data, denoted as $D_{0}$. \\

Next, we find a $\delta$-similarity neighborhood of $D_{0}$, $N_{\delta}(D_{0})=\{\, D_{i}\in S(x)\mid x\in W_{0}, J(W_{0}, W_{i})> \delta \,\}$. Here, $W_{i}$ denotes the list of words of $D_{i}$. If any data fails to be contained in $N_{\delta}(D_{0})$, then it is considered an anomaly. \\

In the above step, we may find some data that are semantically similar but not included in the similarity neighborhood. In such cases, we can take an iterative approach to detect the ultimate anomaly.\\

\textbf{Step 3:} In this step, we first construct the $\delta$-similarity neighborhood of $D_{0}$ with a suitable value of $\delta.$ Suppose there are $m-$ anomalies, say $D^{'}_{1}, D^{'}_{2},...,D^{'}_{m}$. Then, we check $J(W_{0},W^{'}_{i}), i=1,2,...,m$ and take an average of them. Let it be $\delta_{1}$. Now, we construct  a $\delta_{1}$-similarity neighborhood of $D_{0}$ and find anomalies. Thus, repeating the process up to a finite number of times, we will get data $D_{k}^{'}$, for some $k$, for which $J(W_{0}, W^{'}_{k})$ will tend to zero, and in that case $D_{k}^{'}$ will be the ultimate anomaly, where $W^{'}_{k}$ denotes the list of words of $D_{k}^{'}$.\\

Since big data is characterized by the 5 V's—velocity, value, volume, veracity, and variety—the dynamic nature of these features, along with the temporal dimension, poses challenges in presenting concrete examples based on the proposed anomaly detection algorithm. Nevertheless, we provide the Python implementation of our anomaly detection algorithm below. The code has been developed using Python version 3.12.4.\\

 \noindent
 
{\bf Code in Python 3.12.4. :}\\

\noindent
from pyspark import SparkContext, SparkConf\\
from pyspark.sql import SparkSession\\
from pyspark.sql.functions import col, udf\\
from pyspark.sql.types import DoubleType, ArrayType, StringType\\

conf = SparkConf().setAppName("JaccardSimilarity")\\.setMaster("local")\\
sc = SparkContext(conf=conf)\\
spark = SparkSession(sc)\\
data = [
    (0, [ `machine', `learning', `basics']),\\
    (1, [ `deep', `learning', `neural', `networks']),\\
    (2, [ `machine', `learning', `advanced']),\\
    (3, [ `statistics', `data', `analysis']),\\
    (4, [`science', `data', `visualization'])\\
]\\
df = spark.createDataFrame(data, [``id'', ``words''])\\
def jaccard$\_$similarity(list1, list2):\\
    set1, set2 = set(list1), set(list2)\\
    intersection = len(set1.intersection(set2))\\
    union = len(set1.union(set2))\\
    return float(intersection) / union\\
jaccard$\_$udf = udf(jaccard$\_$similarity, DoubleType())\\
reference$\_$doc = [`machine', `learning', `basics']\\
reference$\_$keywords = [`data', `science', `machine', `learning']\\
delta = 0.4\\
broadcast$\_$ref$\_$doc = sc.broadcast(reference$\_$doc)\\
broadcast$\_$ref$\_$keywords = sc.broadcast(reference$\_$keywords)\\
df = df.withColumn(``similarity'', jaccard$\_$udf(col(``words''), spark.create\\DataFrame([(reference$\_$keywords,)], [``words'']).\\select(``words'').first().words))\\
neighborhood$\_$df = df.filter(col(``similarity") > delta)\\
anomalies$\_$df = df.filter(col(``similarity") <= delta)\\
max$\_$iterations = 10\\
current$\_$delta = delta\\
for $\_$in range(max$\_$iterations):\\
    if anomalies$\_$df.count() == 0:\\
        break\\
  avg$\_$similarity = anomalies$\_$df.agg({``similarity":\\ ``avg"}).collect()[0][0]\\
    current$\_$delta = avg$\_$similarity\\
    neighborhood$\_$df = df.filter(col(``similarity") > current$\_$delta)\\
    anomalies$\_$df = df.filter(col(``similarity") <= current$\_$delta)\\
ultimate$\_$anomalies = anomalies$\_$df.collect()\\
print(``Ultimate Anomalies:", [row.words for row in \\ultimate$\_$anomalies])\\
print(``Final Delta:", current$\_$delta)\\
sc.stop()\\

As a case study for our aforementioned anomaly detection algorithm in big data searching, we provide Python code of a case study. This case study is based on customer reviews from an e-commerce platform. For this purpose, we assume that the coder has set up HDFS and that the data is available at `hdfs://path/to/customer/reviews'. Below is the case study-based Python code along with the scenario of the case study:\\

\noindent
\textbf{Scenario}\\
we want to identify anomalous customer reviews in a large dataset from an e-commerce platform. This can help in detecting fake reviews or unusual patterns in the reviews.\\
\noindent
\textbf{Dataset}\\
We assume that one has a dataset of customer reviews stored in HDFS. Each review consists of an ID and a list of words.\\
\noindent
 \textbf{Step 1: Setup Spark and Load Data:} 

    \begin{enumerate}[(a)]
    \item Initialize Spark.
    \item Load the reviews from HDFS.
    \item Broadcast Reference Data:
\end{enumerate}
\noindent

\noindent
\textbf{Step 2: Broadcast Reference Data:}\\
Define and broadcast the reference data and keywords.\\
\noindent
\textbf{Step 3: Calculate Initial Similarity:}\\
Calculate the Jaccard similarity between each review and the reference keywords.\\
\noindent
\textbf{Step 4: Filter Initial Neighborhood and Anomalies: }\\
Identify the $\delta$-similarity neighborhood and anomalies based on the initial delta.\\
\noindent
\textbf{Step 5: Iteratively Adjust Delta: }\\
 Adjust delta based on the average similarity of anomalies and repeat the filtering process.\\
\noindent

Now, we discuss the Python code of the above-mentioned case study.\\

\noindent
from pyspark import SparkContext, SparkConf\\
from pyspark.sql import SparkSession\\
from pyspark.sql.functions import col, udf\\
from pyspark.sql.types import DoubleType\\

\noindent
conf = SparkConf().setAppName("NewsAnomaly\\Detection").setMaster("local[*]")
\\
sc = SparkContext(conf=conf)\\
spark = SparkSession(sc)\\

\noindent
df=spark.read.json("hdfs://path/to/news/articles")\\
\noindent
def jaccard$\_$similarity(list1, list2):\\
    set1, set2 = set(list1), set(list2)\\
    intersection = len(set1.intersection(set2))\\
    union = len(set1.union(set2))\\
    return float(intersection) / union\\
\noindent
jaccard$\_$udf=udf(jaccard$\_$similarity,DoubleType())\\
reference$\_$doc = [`breaking', `news', `headline']\\
reference$\_$keywords = [`breaking', `news', `headline', `today']\\
delta = 0.4\\
broadcast$\_$ref$\_$doc = sc.broadcast(reference$\_$doc)\\
broadcast$\_$ref$\_$keywords = sc.broadcast(reference$\_$keywords)\\
df = df.withColumn("similarity", jaccard$\_$udf(col("words"), \\spark.createDataFrame([(reference$\_$keywords,)], ["words"]).\\select("words").first().words))\\
neighborhood$\_$df = df.filter(col("similarity") > delta)\\
anomalies$\_$df = df.filter(col("similarity") <= delta)\\
max$\_$iterations = 10\\
current$\_$delta = delta\\
for $\_$ in range(max$\_$iterations):\\
    if anomalies$\_$df.count() == 0:\\
        break\\
avg$\_$similarity=anomalies$\_$df.agg({"similarity": "avg"}).\\collect()[0][0]\\
    current$\_$delta = avg$\_$similarity\\
neighborhood\_df = df.filter(col("similarity") > current\_delta)\\
anomalies\_df = df.filter(col("similarity") <= current\_delta)\\
ultimate\_anomalies = anomalies\_df.collect()\\
print("Ultimate Anomalies:", [row['words'] for \\row in ultimate\_anomalies])
print("Final Delta:", current\_delta)\\
sc.stop()\\
\noindent
Now, we also observe an interesting fact that is given below. \\

Let $\mathcal{A}_{1}=\{\, D_{i}\in S(x)\mid J(W_{0},W_{i})\leq \delta_{0}\,\}$ represents the initial set of anomalies. Sometimes, anomalies in the initial set may not be satisfactory due to various reasons, such as semantic differences or the possibility that they represent a similar context to $D_{0}.$ So, we refine our initial threshold value $\delta_{0} $ by using following equation: $$\delta_{1}=\frac{\sum_{D_{i}\in \mathcal{A}_{1}} J(W_{0},W_{i}) }{|\mathcal{A}_{1}|}$$\\ Then our refined set of anomalies becomes $\mathcal{A}_{2}=\{\,D_{i}\mid J(W_{0},W_{i})\leq \delta_{1}\,\}.$\\We repeat this anomaly detection by using updated threshold value by using following iteration equation:\\ $$\mathcal{A}_{n+1}=\{\, D_{i}\in S(x)\mid J(W_{0},W_{i})\leq \delta_{n}\,\}$$ and , $$\delta_{n+1}=\frac{\sum_{D_{i}\in \mathcal{A}_{n+1}} J(W_{0},W_{i}) }{|\mathcal{A}_{n+1}|}$$ here, $$n=0,1,2,......$$\\ Finally, we stop our iteration when $\delta_{n+1}\approx \delta_{n}$ or $J(W_{0},W^{'}_{k})\longrightarrow 0$, for some, $D^{'}_{k}\in \mathcal{A}_{n+2}.$   \\

We call the anomaly detection method described above the $\delta$-method. The value of $\delta$ lies in the interval (0, 1). The initial choice of $\delta$ depends on the choices of the users. Thus, we do not suggest any parameters specifically. Moreover, to avoid arbitrary selection, we propose the iterative refinement of the value of $\delta$, where $\delta$ is dynamically adjusted based on the average similarity of previously identified anomalies. \\

\begin{lemma}
      The collection $\{\, N_{\delta}(D_{0})\mid 0<\delta<1$\,\} is nested.
  \end{lemma}
\begin{proof}
  Let us consider real numbers $\delta_{1},\delta_{2}$, where $0< \delta_{1}< \delta_{2}< 1$. Then, we get $ J(W_{0}, W_{i})> \delta_{2}$ implies that $ J(W_{0}, W_{i})> \delta_{1}$, where $W_{0},W_{i}$  are two lists of keywords of the data $D_{0}$ and $D_{i}$ respectively. Thus, for any $i$, $D_{i}\in N_{\delta_{2}}(D_{0})$ implies $D_{i}\in N_{\delta_{1}}(D_{0}).$ So, $ N_{\delta_{2}}(D_{0})\subseteq N_{\delta_{1}}(D_{0}).$ In similar manner, for reals $\delta_{1}<\delta_{2}<\delta_{3}<...<\delta_{n}$, we have $ N_{\delta_{n}}(D_{0})\subseteq  N_{\delta_{n-1}}(D_{0})\subseteq N_{\delta_{n-2}}(D_{0})\subseteq...\subseteq N_{\delta_{2}}(D_{0})\subseteq N_{\delta_{1}}(D_{0}).$ Hence, the collection $\{\, N_{\delta}(D_{0})\mid 0<\delta<1$\,\} is nested.
\end{proof}

\begin{figure}[H]
    \centering
     \includegraphics[width=6in]{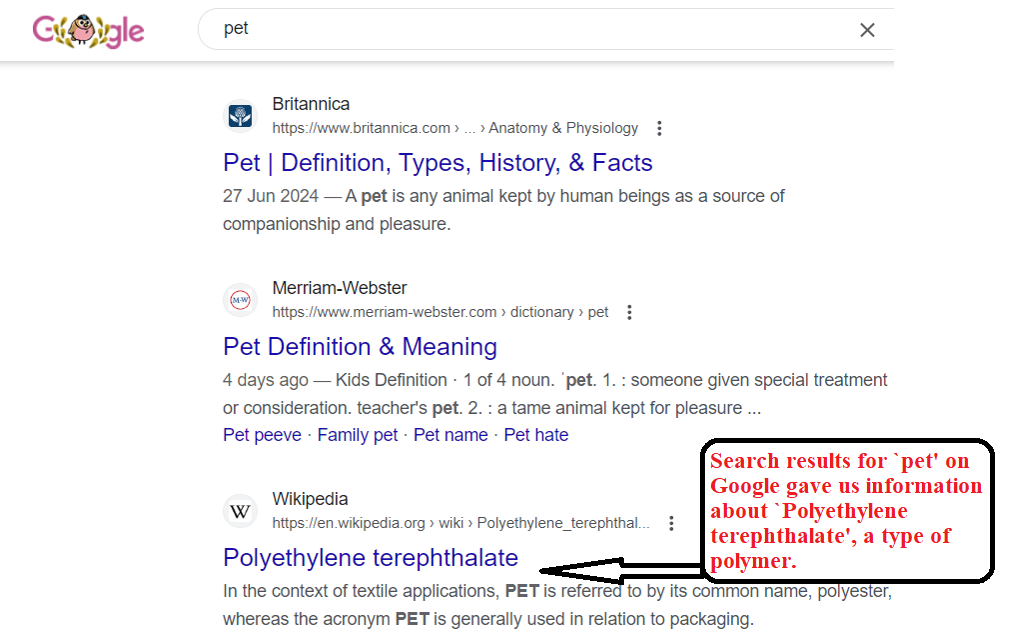}
    \caption{Search result of `pet' on Google containing `Polyethylene terephthalate'. }
    \label{fig: Search result of `pet' on Google containing `Polyethylene terephthalate'.}
 \end{figure}

    \begin{figure}[H]
    \centering
 \includegraphics[width=6in]{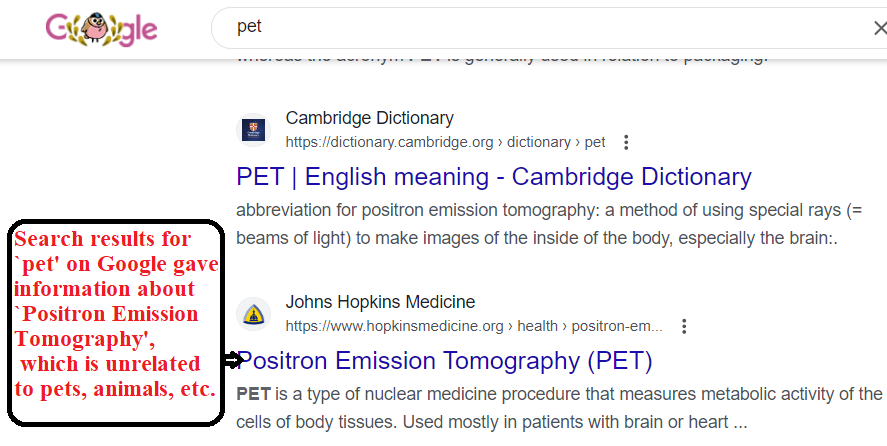 }
 
    \caption{Search result of `pet' on Google containing `Positron emission tomography scan'.}
    \label{fig:Search result of `pet' on Google containing `Positron emission tomography scan'.}
    \end{figure}
    \par

\subsection{Comparison of existing methods for anomaly detection: }

Now, it is important to compare the $\delta$-method, our anomaly detection process described above, with other widely used methods, namely the one-class SVM method \cite{46,47,48} and the isolation forest method \cite{49,50}.\\

\begin{table}[H]
  \centering
  \renewcommand{\arraystretch}{1.6} % Increase row height
  \fontsize{11pt}{13pt}\selectfont  % Set font size
  \begin{adjustbox}{max width=\textwidth}
  \begin{tabular}{|p{0.3\textwidth}|p{0.3\textwidth}|p{0.3\textwidth}|}
    \hline
    \textbf{One-class SVM method} & \textbf{Isolation forest method} & \textbf{$\delta$-method} \\ 
    \hline
    This method mainly deals with  feature vectors, which are numeric in nature \cite{46,48}. While we are to handle text and image data, it requires conversion into numeric format\cite{47}. & The method is not originally intended for text or categorical inputs \cite{50}. It does not model relationships or inclusion among terms, as all the features are treated independently. & It is designed to handle structured textual relations, and it does not require independent numeric transformation of each term. \\
  Since this method relies entirely on numeric feature engineering \cite{46,48}, it can not preserve word ordering or hierarchal structure. & The output is purely based on numeric representation, which lacks word-level interpretability for explaining anomalies. & It is established on a hierarchical search structure, allowing explainability at the word and phrase level. \\ 
    \hline
    This method struggles with context generalization and specification, limiting its effectiveness in varied textual scenarios. & This method does not handle context-based or relational aspects efficiently.  & It can address this using a preorder structure between words.\\ \hline
  \end{tabular}
  \end{adjustbox}
  \caption{Comparison of  the one-class SVM method, isolation forest method, and 
  $\delta$-method.}
  \label{tab:comparison_table}
\end{table}
    
\section{Primal structure in big data search:}
Recently, Acharjee et al. \cite{39} introduced a new notion named `primal' in general topology. Primal is the dual structure of the grill. In this section, we are going to discuss primal structure relating to big data searching in big data analytics.\\
      
From the previous sections, it is clear that in $(B,W,S,\mathbf{R}^{m})$ for a word $x\in W$, its search space $S(x)$ contains all data that contains $x$. For example, the search space $S(World)$
 contains data related to the words
 \textit{`World', `World Health Organization', `World Trade Organization', `World Economic Forum', `World Water Day', `World Map'}, etc. It is noticeable that \textit{S(World Map)} $\subseteq$ \textit{S(World)}, i.e., if we consider $x=$ \textit{`World'}, $y=$ \textit{`Map'}, then $S(x\vee y) \subseteq S(x).$  It is evident that, though in  set theory, a subset of any set contains some of the elements of the set,  in the case of search space $S(x)$, subsets are of the type $S(x\vee y)$. 
\begin{definition}
    Let $B$ be the universe of big data. Then, the collection  $\mathcal{P}\subseteq 2^{B}$ is called a big data primal  if it satisfies the following: 
    \begin{enumerate}[(i)]
        \item $B\notin \mathcal{P},$
        \item if $S(x)\in \mathcal{P} $ and $S(x\vee y)\subseteq S(x),$ then $S(x\vee y)\in\mathcal{P},$
        \item if $S(x)\cap S(y)\in\mathcal{P}$, then $S(x)\in \mathcal{P}$ or $S(y)\in\mathcal{P}.$
    \end{enumerate}
\end{definition}  
  Before going to study big data primal relative to a set of words in the universe of big data $B$, let us consider $M\subseteq W$ be a subset such that  $x, y\in M$ if and only if $x\vee y\in M$. For example, if `Big' and `Data' are in $M$, then `Big Data' is also in $M$ and vice-versa. In $(B, W,S, R)$, we consider a collection $\mathcal{P}_{M}=\{\, S(x)\mid x\in M\,\}.$ In the following part, we discuss that $\mathcal{P}_{M}$ satisfies the definition 8.1 .\\
  
\begin{proposition}
        Let $M\subseteq W$ be any set of words such that  $x, y\in M$ if and only if $x\vee y\in M$. Then, the collection $\mathcal{P}_{M}=\{\, S(x)\mid x\in M\,\}$ is  big data primal (relative to $M$) in the universe of big data $B$.
\end{proposition}
\begin{proof}
\begin{enumerate}[(i)]
    \item The first condition in the definition of big data primal is obvious. Since the set of words $W$ is always finite for an individual, so its subset $M$ is also finite. Thus, it is practically impossible to have the universe of big data in $\mathcal{P}_{M}$. So, $B\notin\mathcal{P}_{M}.$
\item  For any $x,y\in M$, we have $S(x)\in \mathcal{P}_{M}$ and $S(x\vee y)\subseteq S(x)$. Since $x,y\in M$ implies $x\vee y\in M$. Hence, $S(x\vee y)\in \mathcal{P}_{M}.$
\item Let $S(x)\cap S(y) \in \mathcal{P}_{M}$. To show that either $S(x)\in \mathcal{P}_{M}$ or $ S(y)\in\mathcal{P}_{M}.$ Now, $S(x)\cap S(y)\in\mathcal{P}_{M}$ implies $S(x\vee y)\in \mathcal{P}_{M}$. Thus, we have $x\vee y\in M.$ It implies $x\in M, y\in M$. So, $S(x)\in \mathcal{P}_{M}$ or $S(y)\in \mathcal{P}_{M}$.\\

Since $\mathcal{P}_{M}$ satisfies all the conditions stated in definition 8.1, hence, $\mathcal{P}_{M}$ is a big data primal relative to $M$ in the universe of big data $B.$
\end{enumerate}
    \end{proof}
    \begin{example}
        Let $W$ be the set of words. We define a subset $M \subset W$ as:
\[
M = \{ \text{Big}, \text{Data}, \text{Big Data} \}.
\]
Now, $M$ satisfies the condition
$x, y \in M \iff x \lor y \in M$. 
Each element of  $ M$ represents a valid concept. So,  we can  define the big data primal 
$\mathcal{P}_M = \{ S(x) \mid x \in M \}$. 
Here, $S(x)$ is a subset of big data retrieved by word $x$.
    \end{example}
\begin{theorem}
    Let $M,N\subseteq W$ such that $x,y\in M,N\iff x\vee y\in M,N $. In $(B,W,S,R)$, if $\mathcal{P}_{M}$ and $\mathcal{P}_{N}$ two big data primals relative to $M$ and $N$ respectively, then $\mathcal{P}_{M}\cup \mathcal{P}_{N}$ is a big data primal relative to $M\cup N$ in the universe $B$.
\end{theorem}
\begin{proof}
  \begin{enumerate}[(i)]
      \item Given that  $\mathcal{P}_{M} $ and $  \mathcal{P}_{N}$ be two big data primals in $B$. Then, $B\notin \mathcal{P}_{M} $ and $B\notin \mathcal{P}_{N}$. It implies that $B\notin\mathcal{P}_{M}\cup \mathcal{P}_{N}$.
      \item Again, let $S(x)\in \mathcal{P}_{M}\cup \mathcal{P}_{N}$ and $S(x\vee y)\subseteq S(x)$. Now $S(x)\in \mathcal{P}_{M}\cup \mathcal{P}_{N}$ implies $S(x)\in \mathcal{P}_{M}$ or $S(x)\in \mathcal{P}_{N}$. Since  $S(x\vee y)\subseteq S(x)$, so  we have $S(x\vee y)\in\mathcal{P}_{M}$ or $S(x\vee y)\in\mathcal{P}_{N}$. It implies that $S(x\vee y)\in\mathcal{P}_{M} \cup \mathcal{P}_{N}.$
  \item   Let $S(x)\cap S(y)\in \mathcal{P}_{M}\cup \mathcal{P}_{N}.$ To show that $S(x)\in \mathcal{P}_{M}\cup \mathcal{P}_{N}$ or $S(y)\in \mathcal{P}_{M}\cup \mathcal{P}_{N}.$ Now, $S(x)\cap S(y)\in \mathcal{P}_{M}\cup \mathcal{P}_{N}$ implies $S(x)\cap S(y)\in \mathcal{P}_{M}$ or $S(x)\cap S(y)\in \mathcal{P}_{N}.$ But by theorem 3.1 we have, $S(x\vee y)=S(x)\cap S(y). $ It implies that $S(x\vee y)\in  \mathcal{P}_{M}$ or $S(x\vee y)\in \mathcal{P}_{N}.$ Then we have, $x\vee y\in M$ or $x\vee y\in N$. So, $ x,y \in M$ or $x,y\in N$ and thus, $S(x)\in\mathcal{P}_{M}$ or $S(x)\in \mathcal{P}_{N}$. It gives $S(x)\in \mathcal{P}_{M}\cup \mathcal{P}_{N}.$ Similarly, we can show that $S(y)\in \mathcal{P}_{M}\cup \mathcal{P}_{N}.$ Hence, $S(x)\in \mathcal{P}_{M}\cup \mathcal{P}_{N}$ or $S(y)\in \mathcal{P}_{M}\cup \mathcal{P}_{N}.$ \par Therefore, $\mathcal{P}_{M}\cup \mathcal{P}_{N}$ is a big data primal relative to $M\cup N$ in $B$.
  \end{enumerate}

\end{proof}
\begin{corollary}
     Let $M,N\subseteq W$ such that $x,y\in M,N\iff x\vee y\in M,N $. Then, $\mathcal{P}_{M}\cup \mathcal{P}_{N}=\mathcal{P}_{M\cup N}.$
\end{corollary}
    \begin{proof}
        We know that $\mathcal{P}_{M\cup N}=\{\,S(x)\mid x\in M\cup N\,\}.$ Again from theorem 8.1, we have  $\mathcal{P}_{M}\cup \mathcal{P}_{N}$ is big data primal relative to $M\cup N$. Hence $\mathcal{P}_{M}\cup \mathcal{P}_{N}=\{\,S(x)\mid x\in M\cup N\,\}=\mathcal{P}_{M\cup N}.$
    \end{proof}

  {\section{cope of integration with neural networks and deep learning:}

In this section, we discuss the scope of integration of  neural networks and deep learning with our proposed methods. Our proposed topological framework (POBDS) consists mainly of two topologies, $\tau_{F}$ and $\tau_{B}$. So, for each, we can have two different neural network models. In this section, we propose a sketch idea of the $\tau_{B}$-neural network model $(W,R,S,\omega_{R},f)$. Let this neural network consist of $m-$ layers. In each layer $j(j=0,1,2,...,m)$ consists of neurons denoted as $x_{i}^{(j)}$, where $i=1,2,…,m_{j}$. In the input layer $j=0$, we may consider some input words as neurons $x_{i}^{(1)}$, which are atoms, and we consider them for a context-specific search. The next layer, for $j=1$, consists of words from $x\mathcal{R}$, and subsequently the $m$-hidden layer consists of the words from the set $x\mathcal{R}^{m}$. In the next step, we focus on calculating an activation value $y(x^{j})$ for neuron $x^{j}$ by using the formula: $$y(x^{j})=f(\sum_{x^{j-1}\in x\mathcal{R}^{j-1}}\omega_{R}(x^{j-1},x^{j})y(x^{j-1})) $$ here, $f:\mathbf{R}\rightarrow \mathbf{R} $ is an activation function, $\omega_{\mathcal{R}(x,y)} $ is the assigned weight.\\

{\section{Computational complexity:}

\noindent
 
Now, it is important to understand the computational complexity of constructing and manipulating our structures. For this purpose, we assume that $n$ is the number of big data, $m$ is the average number of words per big data, $k$ is the number of reference keywords, and $t$ is the number of iterations in the anomaly detection loop. For the computation of Jaccard Similarity per big data, conversion of lists to sets and computing intersection and union operations each take $\bigO(m + k)$ time. So, for $n$ big data, the total time to calculate Jaccard Similarity is $\bigO(n(m + k))$. On the other hand, the time for computing the average similarity is $\bigO(n)$. So, for a total of $t$ iterations, it requires $\bigO(nt)$ time. So, overall complexity, combining the above, we get $\bigO(n(m + k+t))$. Since the number of big data has been increasing day by day \cite{51}, it is obvious to consider that $m$ and $k$ are relatively small compared to $n$. According to Wynn and Eckert \cite{52}, iterations also cost money. Thus, it is expected that, for efficient big data analytics, $t$ should be relatively small in comparison to $n$. Since $t << n$ along with $m,k << n$, then the overall aforementioned computational complexity reduces to $\bigO(n)$ time.

\section{Comparison with TF-IDF and cosine similarity :}
\noindent
Methodologies, viz., TF-IDF and cosine similarity, are useful in big data analytics to many extents. But our proposed methods of searching for big data using neighborhood structures are different from TF-IDF and cosine similarity. The following table shows the comparison between traditional methodologies, viz., TF-IDF and cosine similarity: 

\begin{table}[H]
  \centering
  {\fontsize{10pt}{10pt}\selectfont
  \renewcommand{\arraystretch}{2.0} % Increases row height
  \begin{adjustbox}{max width=16cm}
  \begin{tabular}{|>{\raggedright\arraybackslash}p{8cm}|>{\raggedright\arraybackslash}p{8cm}|}
    \hline
    \textbf{TF-IDF and cosine similarity} & \textbf{Our proposed methods} \\ 
    \hline
   They recompute each term separately. \cite{32} &   Our methods reuse the search spaces.\\ 
 \hline
They rely on vector space models \cite{33} or term frequency.\cite{32} & Our methods rely on searching hierarchy. \\ 
 \hline
They struggle in context generalization or specification. \cite{32}  &  Our methods overcome this using a preorder structure.\\ 
 \hline
In big data, due to real time data generation, uses of TF-IDF and cosine similarity are computationally expensive. \cite{53} & Our methods do not depend on global document counts.  \\ 
 \hline
  \end{tabular}
  \end{adjustbox}
  }
  \caption{Comparison between TF-IDF and cosine similarity with our proposed methods.}
  \label{tab:comparison-table}
\end{table}

\section{Limitations on empirical and quantitative validations:}

\noindent

We provided several theoretical foundations related to big data search using words in the aforementioned sections. Now, it is our prime responsibility to discuss a few limitations related to empirical and quantitative validations of our theoretical foundations at present. According to IBM \cite{54}, Google \cite{55}, Amazon \cite{56}, and many others, for big data analytics, big data must be available in real time because of volume, which is one of the main characteristics of big data. Also, big data need big storage systems \cite{54,55,56} since they are usually in terabytes, petabytes, or zettabytes in terms of their volumes. Thus, available computing facilities and existing statistical analysis tools are not capable of doing big data analytics \cite{57} since these are suitable for static small data. We use the word `static' since small data are always non-real-time data as well as small in volume \cite{58}. Thus, it is not easy to have empirical or quantitative validations of many theoretical results in big data analytics due to cost, storage systems, computational facilities, and related infrastructures \cite{57,58}. Thus, many researchers have been trying to study big data by considering small data from them \cite{58}. However, since small data are limited in volume, non-heterogeneous, and structured \cite{58}, it should be our concern to consider those small data while doing big data analytics. Moreover, supercomputers are not capable of performing big data analysis \cite{59}. Thus, it is not an easy task to have empirical and quantitative validations of all the theoretical advances of big data analytics at present with the limitations mentioned above. They are only possible when companies like IBM, Amazon, Google, etc., find theoretical advances important enough to implement in their real-time big data or by making minor improvements in theoretical results. One may refer to \cite{58, 60, 61} for detailed comparisons of big data and small data, limitations of data reduction in big data, etc. But in order to overcome these limitations, many experts \cite{57, 62, 63} have been arguing to develop theoretical foundations for big data analytics that are new, novel, and independent of existing theoretical foundations of data science. Moreover, Google, Amazon, Facebook, and Twitter have been using theories to develop techniques of big data analytics as per their own requirements \cite{64}. Thus, we believe that empirical and quantitative validations of this paper may be achieved in the future, which cannot be done at present due to the aforementioned limitations. So, we are unable to show quantitative comparisons with traditional methods such as TF-IDF and cosine similarity with our proposed methods as well as make connections between big data primal structures and their applications in the real world.

\section{Discussion:} 

\noindent

This work introduces foundational topological concepts for analyzing  relationships between words in big data search. Introduction of the preordered big data system (P.O.B.D.S) addresses a structural representation of data points in big data, where theorem 4.1 concludes the idea that search spaces are inherently hierarchical. By introducing forneighborhoods and afterneighborhoods of words,  their associated
topologies $\tau_F$ and $\tau_B$, this paper provides some significant tools in search
optimization that enable hierarchical refinement of search results, where broader terms are considered for more general context, and specific terms are for narrowing down to context-specific searches. The duality of these topologies can be seen in theorem 5.1, which reflects users’
behaviors  because users often start with specific queries and generalize when
results are insufficient, and alternatively, they start broadly first and then refine when looking for specific information. In this work, a layered exploration of search spaces can be seen in definition 5.4, theorem 5.5, and theorem 5.6, where how concepts are evolved as search progresses through increasingly broader or narrower contexts is mentioned. Also, by implementing $m$-steps neighborhoods in definition 5.3, it allows search methods to compute layered
contexts, offering refined or expanded suggestions dynamically. Again,
preservation of the hierarchical structure of search results and consistency across multi-steps searching can be seen in theorem 5.4, theorem 5.3, and their application can be seen in robust multi-step searching algorithms where results derived over multiple iterations retain logical consistency. By introducing ‘Big Data Primal,’
this paper offers a mathematical approach to simplify and organize the searching processes in big data. Proposition 8.1 ensures that for any set of words $M$, the primal structure can be localized to subsets of interest, ensuring that the
framework is adaptable to specific contexts. Applicability of this proposition can be seen in domain-specific searching processes, which enable search methods to specialize in specific subsets of words or topics, such as ‘medical data’ or `financial data’, using big data primal structures. Also, it can be applicable in context-aware recommendations that support personalized
recommendations by focusing on specific sets of user-relevant terms and their relationships. Again, the importance of theorem 8.1 may be seen in combination of knowledge domains, this result demonstrates that primal structures from different domains or topics can be combined seamlessly. These are some of the findings, which enhance big data search methodologies and offer
a flexible framework that can accommodate varying levels of detail and relevance in search results. We provide a concrete example of proposed neighborhood structures along with a pseudo-algorithm to construct them. Also, an example describing the $m$-steps relation is added. A detailed comparison of existing methods for anomaly detection with our proposed $\delta$-method is provided. Again, a scope of integration of our proposed P.O.B.D.S. frameworks with neural networks and deep learning is given. We also calculated the computational complexity of our methods. Further, a detailed comparison of TF-IDF and cosine similarity with our proposed neighborhood structures is included. Moreover, limitations of empirical and quantitative validations are added.   \\

\section{Conclusion:}
In this paper, we investigate hidden topological features in big data analytics that traditional topological data analysis (TDA) cannot study. We establish a preorder relation on the set of words in big data, identifying that the big data search systems operate as preorder big data systems. Utilizing this relation, we introduce new concepts and results related to the forneighborhoods and afterneighborhoods of words within big data. Furthermore, we propose an \textit{m}-steps relation for big data search, which helps us to derive novel topological insights into big data search. Additionally, we introduce the data-directed graph (DDG) and examine some of its properties. This innovative concept opens the way for new discussions on the topological features of big data. We also present a method, named the $\delta$-method,  for anomaly detection in big data search using the Jaccard similarity coefficient.\\

Inspired by the concept of primal, defined by Acharjee et al. \cite{39}, we introduce a generalized version, termed as the big data primal, and explore its properties from the perspective of big data analytics. This big data primal is used to study proximity \cite{65} in big data. Finally, Isham \cite{66,67} established connections between quantum mechanics, lattice theory, and general topology. Given our paper's numerous links to general topology and lattice theory, it is anticipated that our results can be beneficial for future studies on big data from the perspectives of quantum mechanics and quantum computing. Moreover, it is well known that complex systems can be studied using statistical physics, and big data can be generated by these systems \cite{68}. Therefore, our paper may be valuable to experts in different areas.\\

Our proposed topological and relational framework preserves its logical and structural correctness as it is extended to large data sets. Here, the neighborhood structures are based on this preorder structure of words, or more specifically, hierarchical searching. This proposed method never uses any similarity function in the formation of neighborhood. The idea behind locality-sensitive hashing (LSH) \cite{69,70} needs concepts of fuzzy similarity functions and probability \cite{71,72,73}. But our proposed neighborhood structures do not rely on similarity functions and probability. Since our methods have the potential to deal with big data, that is, data in real time, the scalability issue of our methods is dependent on the infrastructures and requirements of the companies, viz., IBM \cite{54, 74}, Amazon \cite{56}, Google \cite{55}, and many others. Thus, our paper may find uses in the future. 

\vskip 0.5cm
\noindent
{\bf Acknowledgment:} The authors are thankful to Prof. Noam Chomsky.

{}
\end{document}